\documentclass[12pt,a4paper,fleqn]{scrartcl}
\RequirePackage{amsthm,amsmath,natbib}
\RequirePackage{amssymb,amstext,pxfonts, dsfont, mathtools}
\usepackage[pdfpagemode=UseOutlines ,plainpages=false
,hypertexnames=false ,pdfpagelabels ,hyperindex=true,colorlinks=true]{hyperref}
	\makeatletter%
	\Hy@breaklinkstrue%
	\makeatother%
\usepackage{color}
\definecolor{darkred}{rgb}{0.6,0.0,0.1}
\definecolor{darkgreen}{rgb}{0,0.5,0}
\definecolor{darkblue}{rgb}{0,0,0.5}
\hypersetup{colorlinks ,linkcolor=darkblue ,filecolor=darkgreen
,urlcolor=darkblue ,citecolor=black}
\usepackage{ bm}
\usepackage{bbm}

\renewcommand{\cite}{\citet}
\usepackage{graphicx}
\usepackage{pgf, tikz}
\usepackage{enumerate}
\usepackage{subcaption}
\usetikzlibrary{positioning}

\def\Ex{\mathop{{\mathbf E}\null}\nolimits}%
\def\1{\mathop{\mathbbm 1}\nolimits}
\DeclareMathOperator*{\argmin}{arg\,min}

\newcommand{\PP}{\mathbb{P}}
\DeclarePairedDelimiter{\mynorm}{\lVert}{\rVert}

\definecolor{dgreen}{rgb}{0,0.5,0}
\definecolor{dblue}{rgb}{0,0,0.9}
\definecolor{dred}{rgb}{0.6,0.0,0.1}
\definecolor{dgold}{rgb}{0.5,0.3,0.0}
\definecolor{dvio}{rgb}{0.6,0.3,0.5}
\definecolor{gray}{rgb}{0.5,0.5,0.5}

\usepackage{setspace}

\theoremstyle{mysc}\newtheorem{prop}{Proposition}[section]
\theoremstyle{mysc}\newtheorem{coro}[prop]{Corollary}
\theoremstyle{mysc}\newtheorem{theo}[prop]{Theorem}
\theoremstyle{mysc}
\theoremstyle{mysc}\newtheorem{lem}[prop]{Lemma}
\theoremstyle{myex}\newtheorem{rem}{Remark}[section]
\theoremstyle{myex}
\theoremstyle{myex}

 \theoremstyle{mysc}\newtheorem{assA}{Assumption}

\newtheorem*{assumption*}{\assumptionnumber}
\providecommand{\assumptionnumber}{}
\makeatletter

\numberwithin{equation}{section}

  \author{{\textsc{ Stephan Martin}}
  \thanks{Deutsche Bundesbank; Frankfurt, Germany. The paper represents the author's personal opinion and does not necessarily reflect the views of the Deutsche Bundesbank or its staff. Humboldt-Universit\"at zu Berlin, Spandauer Stra\ss e 1, 10178 Berlin, Germany. Financial support by the Collaborative Research Center TRR 190 "Rationality and Competition" is gratefully acknowledged. 
  e-mail:
\url{	stephan.martin@bundesbank.de}} \\
{\small \textit{  Deutsche Bundesbank / Humboldt-Universit\"at zu Berlin}}
}

 \title{Estimation of Conditional Random Coefficient Models using Machine Learning Techniques
 }

\begin{document}
   \maketitle
\begin{abstract}
Nonparametric random coefficient (RC)-density estimation has mostly been considered in the marginal density case under strict independence of RCs and covariates. This paper deals with the estimation of RC-densities conditional on a (large-dimensional) set of control variables using machine learning techniques. The conditional RC-density allows to disentangle observable from unobservable heterogeneity in partial effects of continuous treatments adding to a growing literature on heterogeneous effect estimation using machine learning. 
This paper proposes a two-stage sieve estimation procedure. First a closed-form sieve approximation of the conditional RC density is derived where each sieve coefficient can be expressed as conditional expectation function varying with controls. Second, sieve coefficients are estimated with generic machine learning procedures and under appropriate sample splitting rules. The $L_2$-convergence rate of the conditional RC-density estimator is derived. The rate is slower by a factor then typical rates of mean regression machine learning estimators which is due to the ill-posedness of the RC density estimation problem. The performance and applicability of the estimator is illustrated using random forest algorithms over a range of Monte Carlo simulations and with real data from the SOEP-IS. 
Here behavioral heterogeneity in an economic experiment on portfolio choice is studied. The method reveals two types of behavior in the population, one type complying with economic theory and one not. The assignment to types appears largely based on unobservables not available in the data.   
\end{abstract}
\section{Introduction}
In recent years microeconometric models aimed at capturing complex heterogeneity in individual behavior. 
In particular, it has become relevant to include observed and/or unobserved heterogeneity when modeling (average) partial effects in regression models. 

An important model that accounts very flexibly for such heterogeneity is the nonparametric random coefficient model
\begin{equation*}
Y = B_{0} + B_{1}W,
\end{equation*}
where $(B_{0}, B_{1})$ is a vector of $p+1$ random variables and $W \in \mathbb{R}^p$ is a vector of regressors. Let $X$ be a set of additional control variables that may be related to $(W,B_{0}, B_{1})$.
$B_{1}$ is the individual effect of a change in $W$ on $Y$ and the distribution of $B_1$ reflects the heterogeneity of individual effects in the population. 

The model is nonparametric in that it does not impose distributional assumptions on partial effects $B_1$ and the nuisance $B_0$. The primary goal for this class of models is to identify and estimate the entire distribution of the vector of random coefficients $(B_0, B_1)$ such as the joint density function. The marginal density of the random slope parameter $B_1$ is of special interest in economic applications as $B_1$ can be interpreted as an average partial effect of a change in $W$ on the outcome. Its distribution reflects the heterogeneity of this effect in the underlying population. If $W$ is an exogenous treatment then the expected value of $B_1$ corresponds to an average treatment effect, its conditional expectation to an conditional average treatment effect and so on.

A crucial identifying assumption in the literature is full independence of covariates $W$ and random coefficients $(B_0,B_1)$.
This is satisfied if $W$ is randomly assigned, such as in experimental data settings, and the marginal slope density captures the entire heterogeneity of a partial effect in the population. This knowledge does not allow to link the heterogeneity to any observable characteristics. For instance, the shape of the RC-density may vary across observable individual characteristics. Learning the RC-density conditional on a set of control variables $X$ provides additional insight on how the shape of the heterogeneity varies across subpopulations with differing observable characteristics. This allows to disentangle heterogeneity in \textit{observable} and \textit{unobservable} heterogeneity.

Furthermore, when dealing with observational data there is always room for a potential dependence between $W$ and control variables $X$. The random intercept $B_0$ subsumes the effects of $X$ on $Y$ which violates full independence between $W$ and random coefficients. 
In this work the identifying restriction can be weakened to allow for conditional independence, i.e. $B \Perp W |X$, which corresponds to a selection-on-observables assumption.  
In most economic applications the set of control variables $X$ is of considerable size or even high-dimensional and there is generally no prior knowledge about which variables in $X$ drive heterogeneous partial effects. Modern Machine Learning (ML) methods allow to deal with large dimensional set of controls and perform some form of variable selection to identify those elements of $X$ inducing different shapes of heterogeneity. 
Recently ML-techniques have proven useful in relating features of the distribution of partial effects $B_1$ to additional observable characteristics, such as in the estimation of conditional average treatment effects, also referred to as heterogeneous treatment effect.  
In this work, I link the entire distribution of random coefficients to observable characteristics by studying a \textit{conditional random coefficient model}. This allows to uncover (i) which variables generally drive heterogeneity in partial effects and (ii) how the distribution of partial effects varies across subpopulations with different observable characteristics. \\

First I begin by providing an identification statement for the conditional RC-density. Deriving from this identification statement I can formulate a sieve approximation to the RC-density conditional on a fixed value of $X$. This sieve approximation has a closed form expression and each sieve coefficient can be expressed as a conditional expectation function of some nonlinear transformation of $Y$ and $W$ varying with controls $X$. 

Generic ML-methods can be used to estimate this set of conditional expectation functions.
Various practical considerations to be outlined later require to orthogonalize both the outcome $Y$ and treatment $W$ using ML technqiues before estimating the sieve coefficients itself.  
 This requires the use of iterated sample splitting to deal with nested ML-steps. 

Following the outline of the estimation strategy I derive the $L_2$-convergence rate of the final conditional RC-density estimator. There convergence rate is crucially determined by the asymptotic properties of the ML-methods employed. 
Given that the slowest ML-estimator employed converges at a polynomial rate the $L_2$-convergence rate of the RC-density estimator is by a factor slower than that of the slowest ML-estimator. This factor hinges on the degree of ill-posedness of the underlying random coefficient problem and the overall smoothness of the density. 

In addition to the asymptotic properties of the estimator, I introduce a cross-validation strategy to inform the choice of tuning parameters. 

Finally I apply the estimator to study behavioral heterogeneity in an economic experiment on portfolio choice. Survey respondents of the german socio-economic panel (SOEP) are asked to invest an hypothetical monetary amount into a riskfree asset or into an risky asset with payoff depending on the return of a stock market index. I study the effect of stock market beliefs on the investment decision with a random coefficient model. I find a bi-modal RC-density that reflects the presence of two types in the population. One type complies with economic theory and the stock market beliefs have a positive impact on the amount invested in the risky asset. For a second type this is not the case and the effect is centered around zero. This type-division prevails when varying the values of controls. This suggests that the assignment to types is largely based on unobservable characteristics not in the data. An exemption is age, as for a subpopulation of elder respondents the share of non-compliers is substantially larger.     


\paragraph{Related literature}
Identification and estimation of the nonparametric random coefficient model is studied in \cite{beran1992}, \cite{beran1994}, \cite{beran1996} and \cite{hodermam_RC_2010}. See also \cite{masten17_RC} for a refined identification result. 
All of these works operate under the assumption that random coefficients and regressors are fully independent. \cite{Breunig_VRC} considers a so called \textit{varying random coefficient model} where each random coefficient is made up of a nonparametric function of control variables and an additively separable random component which is fully independent of controls. Further the number of control variables affecting random coefficients is effectively smaller than the number of random coefficients itself whereas in the present setting the number of controls is allowed to be larger.  
Sieve estimation for random coefficient models is used for the testing procedure in \cite{BreunigHoder18} and in \cite{Breunig_VRC}.  

Machine Learning estimation and econometric models have been paired frequently in recent years. 
\cite{ChernPostSelec}, \cite{ChernDebiased17} and \cite{chernozhukov2020debiased} study estimation and inference on parameters and linear functionals of parameters in high-dimensional linear models. Thereby highlighting the importance of iterated ML estimation and sample splitting for achieving consistency and asymptotic normality of parameter estimates. 
  
Of particular importance for this work is the estimation of heterogeneous treatment effects using machine learning methods as considered in \cite{AtheyWager18} and \cite{grf}, see also the references therein. This is due to the fact that a heterogeneous treatment effect can be viewed as the conditional expectation function of a random slope coefficient. In contrast, the conditional RC-density studied here is informative about the entire distribution of a treatment effect in a given (sub-)population. This also includes learning the form of \textit{unobservable} heterogeneity which remains otherwise unknown when only the mean of a random coefficient is studied. 
\cite{ChernGenericML} considers identification and estimation of particular subfeatures of the conditional expectation function of a random slope.

The theory of the conditional RC density estimate developed here holds for generic machine learning techniques. However I make use of the causal forest algorithms of \cite{grf} and the popular ML-tool of random forests introduced by \cite{Breiman01} in the implementation of the estimator and in the asymptotic theory. The asymptotic theory of random forest estimators has been studied in \cite{scornet2015}, \cite{WagerWalther}, \cite{AtheyWager18} and \cite{grf}.   
 \\


The remainder of this paper is organized as follows. Section \ref{Sec::ModelSetup} introduces the main model and discusses identification of conditional RC-densities. Section \ref{Sec::Est} outlines the estimation strategy. Section \ref{Sec::Asymp} presents the asymptotic properties of the estimator and Section \ref{Sec::Inference} asymptotic inference on the conditional RC density estimates. Section \ref{Sec::CVetc} addresses auxiliary topics, i.e. marginal density estimation, variable importance measures and the choice of tuning parameters. Section \ref{Sec::MCs} contains a Monte Carlo simulation study. Section \ref{Sec::Application} contains an empirical application of the estimation procedure using survey data. Section \ref{Sec::Conclusion} concludes.

\section{Model Setup and Identification}\label{Sec::ModelSetup}
This paper considers the following random coefficient model
\begin{equation}\label{RC_Model}
Y = B_0 + B_1\cdot W
\end{equation}
where $B=(B_0, B_1)$ consists of two scalar random variables and $W$ is a scalar regressor of interest to the researcher. Here $B_1$ is the average partial effect of a change in $W$ on the outcome $Y$. Within the model framework there exists a vector of additional covariates $X \subseteq \mathbb{R}^d$ that may affect both the random coefficients as well as the regressor $W$.
The goal of this section is to give conditions under which we can identify the conditional random coefficient density $f_{B|X=x}$. That is the random coefficient density for a given subpopulation with characteristics $X=x$. Note that the results of this paper can be readily extended to cover the case where $W$ is multivariate, though for conciseness I focus on the scalar $W$ case which is of the most practical relevance. 
    
The model (\ref{RC_Model}) can be interpreted as a reduced form of a more general multivariate random coefficient model with the random intercept absorbing all the (heterogenous) effects of other covariates on the outcome.  
Without loss of generality the random coefficients satisfy
\begin{equation*}
B_j = g_j(X) + A_j, \;\;\text{where}\;\; \Ex[A_j\;|\; X]=0, \;\; j=0,1
\end{equation*} 
which illustrates that (\ref{RC_Model}) nests many popular mean regression models. For instance it can be viewed as an extension of \cite{Robinson88} with a random coefficient instead of a deterministic one.
\\
For obtaining identification the following assumptions are imposed.
\begin{assA}\label{A::Ident1}
	(i) $B \Perp W \;|\; X$ (ii) for every $x$ in the support of $X$ the random variable $W \;|\; X=x$ has full support $\mathbb{R}$.
\end{assA}
Assumption \ref{A::Ident1} (i) requires that the regressor of interest $W$ is independent of random coefficients conditional on controls $X$. This restriction allows for some dependence between $B$ and $W$ and is thereby weaker then the full independence assumption typically encountered in random coefficient models, see \cite{beran1992}, \cite{hodermam_RC_2010} and \cite{masten17_RC}. It can also be interpreted as an exogeneity condition on the regressor $W$. If $W$ is a (quasi-)experimental intervention then (i), or more precisely the part $B_0 \Perp W \;|\;X $, corresponds to a selection on observables assumption. It is also one of the main assumptions for identifying heterogenous treatment effects as in the RC model in Section 6 of \cite{grf}.

Assumption \ref{A::Ident1} (ii) strengthens the common large support restriction from the random coefficients literature (see again e.g. \cite{hodermam_RC_2010}) to also hold conditional on $X=x$.
The assumption rules out the case where $W$ is a deterministic function of $X$ and may be problematic if otherwise certain realizations $x$ provide strong information about $W$. A workaround for this assumptions is provided by the varying RC model of \cite{Breunig_VRC}. There however functional form restrictions on the random coefficients need to be made which outlines a tradeoff between the above full support assumption and further restrictions on the random coefficient model.      
If this assumption is violated for some realizations of $X$ then nevertheless identification for different $x$-values which satisfy (ii) can be established.
\cite{masten17_RC} discusses identification in the case of bounded support of the regressor $W$. Taking this results in to account the following identification result can be formulated.       
\begin{lem}\label{Lem::Ident}
	If Assumption \ref{A::Ident1} holds then for every $x$ in the support of $X$ the density function $f_{B|X=x}$ is identified. 
	If $W \;|\; X=x$ has instead only compact support then $f_{B|X=x}$ is point identified if and only if the distribution of $B\;|\;X=x$ is determined solely by its moments and all absolute moments are finite.
\end{lem}
The proof of Lemma \ref{Lem::Ident} immediately follows from extending classical identification results for random coefficient models as synthesized in \cite{masten17_RC} to the conditional case. 
This section concludes with a remark on identification for the case where $W$ is a discrete, binary variable.
\begin{rem}
	The binary treatment case $W \in \{0,1\}$ is often more relevant in practical applications and identifying heterogenous average partial effects in a binary treatment model has received much attention recently, see e.g. \cite{AtheyWager18}.
	  
	  If $W$ depends on controls $X$, i.e. there is selection on observables into treatment, we can reformulate \ref{RC_Model} to 
	  \begin{align*}
	  Y =\widetilde B_0 + B_1\cdot (W-\Ex[W|X]), \;\;\text{with}\;\; \widetilde B_0=B_0+ B_1\cdot \Ex[W|X]
	  \end{align*}
	  and the orthogonalized regressor $W-\Ex[W|X]$  will be supported on a compact interval that is a subset of $[-1, 1]$, given the well known overlap assumption\footnote{For any $x$ in the support of $X$ there exists an $\epsilon > 0$ such that $\epsilon < \PP(W=1|X=x) < 1- \epsilon$ holds.} is satisfied. 
	  Then Lemma \ref{Lem::Ident} can be applied to obtain identification of the distribution $\tilde{B}_0, B_1 |X=x$ within the constraints set in Lemma \ref{RC_Model}. From this identification of the economically relevant random slope density $f_{B_1|X=x}$ can be inferred.
\end{rem}

\section{Estimation of Conditional Random Coefficient Densities}\label{Sec::Est}
This section introduces the estimation strategy for conditional RC-densities. Subsection \ref{Sec::EstIdea} outlines the principal idea and introduces the main notation whereas Subsection \ref{Sec::Demean} is of main practical relevance. There a demeaned random coefficient model is discussed and further details on machine learning estimators and sample splitting are provided.
Before moving forward the following general notation needs to be introduced.    
Let
\begin{align*}
\phi_{Y|X}(t|x)=\Ex[\exp(itY)\;|\; X=x]
\end{align*}
denote the characteristic function of $Y$ conditional on $X=x$. 
The Fourier transformation $\mathcal{F}$ and the inverse Fourier transformation $\mathcal{F}^{-1}$ are defined as
\begin{align*}
(\mathcal{F}f)(t) & =\int_{\mathbb{R}^d} \exp(it'a)f(a)da \\
(\mathcal{F}^{-1}g)(a) &= \frac{1}{(2\pi)^d}\int_{\mathbb{R}^d} \exp(-ia't)g(t)dt 
\end{align*}
for some functions $f,g: \mathbb{R}^d \to \mathbb{R}$. The operators $\mathcal{F}: \mathbb{R}^d \to \mathbb{C}^d$ and $\mathcal{F}^{-1}: \mathbb{C}^d \to \mathbb{R}^d$ relate the characteristic function of a random variable to its probability density function, given the latter exists. For some random variable $A$ with density $f_A$ it holds that $\phi_{A}(t) = (\mathcal{F}f_{A})(t)$ and vice versa $(\mathcal{F}^{-1}\phi_{A})(a)= f_{A}(a)$.

\subsection{A Two-Stage Sieve Estimation Approach}\label{Sec::EstIdea}
The essential implication of Assumption \ref{A::Ident1} that will be leveraged for estimation is the identity 
\begin{align*}
(\mathcal{F}f_{B|X=x})(t,tw) = \phi_{Y|X, W}(t|x,w)
\end{align*}
which holds for every $x$ in the support of $X$ and $t,w \in \mathbb{R}$. The identity in turn implies the following $L_2$-condition
\begin{align}\label{L2cond}
\int_{\mathbb{R}^2} \Big|(\mathcal{F}f_{B|X=x})(t,tw) - \phi_{Y|X, W}(t|x,w) \Big|^2 d\nu(t)d\mu(w) = 0
\end{align}
where $v,\mu$ are arbitrary probability measures on $\mathbb{R}$ that are discussed later in more detail.   
Following \cite{Breunig_VRC} the $L_2$-criterion in (\ref{L2cond}) can be used to construct a sieve estimator of the density $f_{B|X=x}$.  

To this end let $q^K = (q_1,\dots, q_K)$ denote a $K=K(n)$-dimensional vector of known basis functions that span the linear sieve space $\mathcal{B}_K=\{\phi(\cdot)=q^K(\cdot)'\pi \}$. As $f_{B|X=x}$ is bivariate, $q^K$ typically is a tensor product of univariate basis functions, i.e. $q^K(b_0, b_1)=q^{K_1}(b_0) \otimes q^{K_2}(b_1)$ with $K=K_1 \cdot K_2$. 

If the characteristic function $\phi_{Y|X, W}(\cdot|x,w)$ were known then a sieve estimator of the conditional random coefficient density is
\begin{align*}
\widetilde{f}_{B |X}(\cdot |x) = \arg \min_{\phi \in \mathcal{B}_K} \int_{\mathbb{R}^2} \Big|(\mathcal{F}\phi)(t,tw) - \phi_{Y|X, W}(t|x,w) \Big|^2 d\nu(t)d\mu(w) 
\end{align*}
which has the following closed form expression
\begin{align}\label{Est::closed}
\widetilde f_{B | X}(b |x)= q^K(b)'Q^{-1}   \int_{\mathbb{R}^2}  (\mathcal{F}q^K )(-t, -tw) \phi_{Y|X, W}(t|x,w) d\nu(t)d\mu(w) 
\end{align}
and where
\begin{align}
Q = \int_{\mathbb{R}^2}  (\mathcal{F}q^K )(t, tw) (\mathcal{F}q^K )'(-t, -tw)d\nu(t)d\mu(w).
\end{align}
Note that the estimator in (\ref{Est::closed}) is not feasible as $\phi_{Y|X, W}(t|x,w)$ is not known. \cite{Breunig_VRC} proceeds by replacing the unknown characteristic function with a nonparametric plug-in estimate. A general problem is the presence of the possibly large-dimensional set of controls $X$ that cannot be reduced a priori in most practical applications. \cite{Breunig_VRC} studies a varying random coefficient model which puts additional structure on the relationship of random coefficients and controls and where $X$ is low-dimensional.

Modern Machine Learning (ML-)estimators are well suited for the estimation of conditional expectation functions like $ \phi_{Y|X, W}(t|x,w)=\Ex[\exp(itY) |X=x, W=w]$ in the presence of possibly high-dimensional controls $X$. However an additional issue arises in that we would require to perform different Machine Learning steps over a continuum of values for $t$.

In order to enable the use of Machine Learning techniques a different strategy is needed. 
Further rearranging of (\ref{Est::closed}) yields 
\begin{align}\label{Est::closed_mu}
\widetilde f_{B | X}(b |x)= & q^K(b)'Q^{-1}   \int_{\mathbb{R}} \int_{\mathbb{R}}  \Ex[ (\mathcal{F}q^K )(-t, -tW)\exp(itY) | X=x, W=w ]  d\nu(t)d\mu(w) \nonumber \\ 
= &  q^K(b)'Q^{-1} \int_{\mathbb{R}}  \Ex[\int_{\mathbb{R}} (\mathcal{F}q^K )(-t, -tW)\exp(itY)d\nu(t) | X=x, W=w ]d\mu(w) \nonumber \\ 
= &q^K(b)'Q^{-1}  \int_{\mathbb{R}}  \Ex[T(W, Y) | X=x, W=w ]d\mu(w) 
\end{align}
where the operator $T(w, y) := \int_{\mathbb{R}} (\mathcal{F}q^K )(-t, -tw)\exp(ity)d\nu(t)$ is short-hand for the nonlinear mapping $ T: \mathcal{W}\times \mathcal{Y} \to \mathbb{R}^K $. The operator $T$ can be computed via numeric integration and thus I consider it deterministic for the remainder of this paper. Also for clarification let $T(W,Y)=(T_1(W, Y), \dots, T_K(W, Y))$ with functions $T_k:\mathcal{W}\times \mathcal{Y} \to \mathbb{R}$, for $k=1,\dots, K$. Further define the functions 
\begin{align*}
\pi_k(x, w)=\Ex[T_k(W,Y)|X=x, W=w]
\end{align*}  
for $k=1,\dots, K$ with $\pi(x,w)=(\pi_1(x,w), \dots, \pi_K(x,w)$.
A distinguishing feature of the sieve approximation in (\ref{Est::closed_mu}) is that sieve coefficients are the relevant quantity that varies in $X$. 

Finally I construct a feasible estimator by replacing $Q$ with a sample mean and $\pi(x,w)$ with a vector of Machine Learning estimates. The resulting RC-density estimator is
\begin{align}\label{Est::closed_no_mu}
\widehat f_{B | X}(b |x) = q^K(b)'\widehat Q^{-1} \int_{\mathbb{R}} \widehat \pi(x,w)d\mu(w)
\end{align}
where $\widehat \pi$ is a generic ML-estimate of the unknown function $\pi$ and 
\begin{align*}
\widehat Q = \frac{1}{n} \sum_{i=1}^{n} (\mathcal{F}q^K )(N_i, N_iM_i) (\mathcal{F}q^K )'(-N_i, -N_iM_i)
\end{align*}
where the $(M_i, N_i)$'s are $n$ Monte Carlo draws from the probability distributions $\mu$ and $\nu$ that are specified by the researcher. In principal $Q$ can be calculated directly via numerical methods for most measures $\mu, \nu$, however this representation is introduced here, as we will later consider the case where $\mu$ is the distribution of $W$ and use sample realizations $W_i$ instead of the generated $M_i$.

Notice that (\ref{Est::closed_no_mu}) is a direct estimate of the closed-form sieve projection in (\ref{Est::closed}). The sieve coefficients can be expressed in terms of different conditional expectation functions which can in turn be conveniently estimated by generic machine learning routines even if the set of controls $X$ is high-dimensional.  

\paragraph{Choice of weighting measures}

The choice of the measure $\mu$ leaves room for further simplification of the estimator in (\ref{Est::closed_no_mu}). If we choose $f_{W|X=x}$ as the density of the weighting measure $\mu$ then it holds that $\int_{\mathbb{R}} \pi(x,w)d\mu(w) = \Ex[T(W,Y)|X=x]$ and we only need to consider ML-estimation of a conditional expectation function and there is no need for additional weighting of this ML-estimator. This does slightly ease computation and simplifies the asymptotic analysis as asymptotic properties of ML-estimators of conditional expectation functions are readily available. The properties of the transformed ML-estimator $\int_{\mathbb{R}} \widehat \pi(x,w)d\mu(w)$ used to calculate the estimator in (\ref{Est::closed_no_mu}) have not been studied explicitly.

A caveat of choosing $f_{W|X=x}$ is that the matrix $Q$ will vary in $x$ which is problematic from both the computational as well as the theoretical stance\footnote{Each of the $K^2$ elements of $Q$ would need to be estimated by a ML-step. Further the asymptotic properties of a machine-learned matrix whose dimensions increase with the sample size remain unclear.}. A possible workaround is to orthogonalize $W$ and is detailed in the next section. However, this workaround will slow down the convergence of the RC-density estimator, as will be shown in Section \ref{Sec::Asymp}. Hence choosing $d\mu/dw=f_{W|X=x}$ as weighting measure is problematic. 

In any case choosing a $\mu$ that is related to the distribution of $W$ is appropriate. ML-estimators like $\widehat{\Pi}(x,w)$ perform best for points from the center of the distribution and will be less accurate in the tails. Moving forward I will focus on the case where $d\mu(w)/dw=f_W(w)$ which automatically weighs down areas where the ML-estimates may be less accurate. This respects the finite sample behavior of each machine learned sieve coefficient and reduces the problem of estimating $Q$ to a simple sample mean. The additional weighting of the ML-estimate $\widehat{\pi}(x,w)$ is a minor issue compared to the difficulties arising from alternative choices for $\mu$. \\

Following \cite{Breunig_VRC}, I choose $\nu$ to follow a $\text{lognormal}(0, \sigma_t)$ distribution where $\sigma_t > 0$ is then the second tuning parameter to be chosen by the researcher, along with $K$.  
This particular choice of weighting measure works well in the settings of \cite{Breunig_VRC} and also in the simulations and applications in this work. Theoretical justifications are given in \cite{BreunigHoder18}, \cite{Breunig_VRC} and in Section \ref{Sec::Est} but these do not preclude other choices of weighting distributions.  
\paragraph{Choice of sieve basis functions}
Throughout this paper I again follow \cite{Breunig_VRC} and choose Hermite functions as sieve basis $q^K$.

Hermite functions are a $L_2$-basis and have appealing theoretical properties. 
They are eigenfunctions of the Fourier transformation and satisfy $\mathcal{F}q_k(a,b)=\sqrt{2\pi}i^{k-1}q_k(a,b)$. This property simplifies the computation of the estimator considerably. 
A drawback of Hermite functions is that most of the support concentrates around zero even if $K$ is moderately large. Thus any moderately sized sieve approximation will fail to be a good approximation of a density function that is centered away from zero and/or has a particularly large support. This is a major motivation for considering a demeaned random coefficient model in the next subsection.\\

\subsection{Demeaning of Random Coefficients}\label{Sec::Demean}
This section discusses estimation of a demeaned version of the random coefficient model in (\ref{RC_Model}). 
The reason is that if marginal densities of $B_0$ and $B_1$ are centered away from zero, estimation of the bivariate density function with a Hermite function sieve will require a possibly large choice of $K$ and thus prohibitively many ML-steps.  
The computational cost associated with each ML-step and the general ill-posedness of the RC-density estimation problem lead to a strong preference for a coarse choice of $K$. 

Second, as the random slope density is of particular interest in economics, specific ML-routines have already been developed which provide high quality estimates for the conditional expectation $\Ex[B_1 |X]$, see \cite{AtheyWager18} and \cite{grf}. By demeaning we can separate estimation of the conditional expectation from the remaining conditional shape of the RC-density. This enables the use of ML-tools that are tailored to the specific predictive tasks such as the estimators in \cite{grf} for conditional expectation function of $B_1$. Further this direct estimators of the conditional expectation will perform better than those one can infer from an indirect estimate via integrating the entire conditional density function. \\  

We can reformulate the original RC-model (\ref{RC_Model}), 
\begin{align*}
Y-\Ex[Y|X,W] & = A_0 + A_1\cdot W, \\
A_0 &=  B_0 - \Ex[B_0|X]  \;\;\text{and}\;\; A_1 = B_1 - \Ex[B_1|X]   
\end{align*}
and estimate the joint density of the demeaned random coefficients $f_{A|X=x}$ with the procedure outlined in the previous section. To this end let $\beta(x)=\Ex[B|X=x]$ with $\beta(x)=(\beta_0(x), \beta_1(x))$ denoting the conditional expectation of the intercept and slope respectively. Further let $m(x,w)=\Ex[Y|X=x, W=w]$.
Then the closed form of the sieve approximation analogous to the previous section is
\begin{align}\label{Est::SieveApproxDemean}
\widetilde f_{B | X}(b |x) &= \widetilde f_{A | X}(b - \beta(x)  |x) \nonumber \\
& = q^K(b - \beta(x))'Q^{-1}\int_{\mathbb{R}}\Ex[T(W, Y- m(X,W )) | X=x, W=w]d\mu(w) 
\end{align}
with 
\begin{align*}
Q = \int_{\mathbb{R}^2}  (\mathcal{F}q^K )(t, tw)(\mathcal{F}q^K )'(-t, -tw)d\nu(t)d\mu(w).
\end{align*}
Further define
\begin{align*}
\Pi(x)& =\int_{\mathbb{R}}\Ex[T(W, Y- m(X,W)) | X=x, W=w]d\mu(w) \\
\Pi_{dm}(x)& =\int_{\mathbb{R}}\Ex[T(W, Y- \widehat m(X,W )) | X=x, W=w]d\mu(w) 
\end{align*}
with $\widehat m$ denoting a generic (ML)-estimator for the unknown function $m$.
By choosing $d\mu(w)/dw=f_W(w)$ an estimator is
\begin{align}\label{Final_Estimator}
\widehat f_{B | X}(b |x) &= \widehat f_{A | X}(b - \widehat \beta(x)  |x) \nonumber \\
& = q^K(b - \widehat  \beta(x))'\widehat Q^{-1}\widehat{\Pi}_{dm}(x)
\end{align}
where 
\begin{align}\label{Est Q}
\widehat Q = \frac{1}{n} \sum_{i=1}^{n}  (\mathcal{F}q^K )(t, t\cdot W_i) (\mathcal{F}q^K )'(-t, -t\cdot W_i)d\nu(t)
\end{align} 
and $\widehat{\Pi}_{dm}(x)$ is an ML-estimate of $\Pi_{dm}(x)$. The conditional expectation $\Ex[T(W, Y- \widehat m(X,W )) | X=x, W=w]$ can be conveniently estimated with ML methods, but it remains to construct an estimate for the quantity $\int_{\mathbb{R}}\Ex[T(W, Y- \widehat m(X,W )) | X=x, W=w]d\mu(w)$. I suggest to use 
\begin{equation}\label{sample_mean_trafo}
	\widehat{\Pi}_{dm}(x) = \frac{1}{R}\sum_{r=1}^{R} \widehat \Ex[T(W, Y- \widehat m(X,W )) | X=x, W=W_r]
\end{equation}
where $\widehat \Ex[T(W, Y- \widehat m(X,W )) | X=x, W=w]$ is an ML-estimator of the respective conditional expectation function and thus $\widehat{\Pi}_{dm}(x)$ is a sample average of different predictions from the ML-estimators over a hold-out sample of size $R$ which has not been used in the estimation before.
Another possibility, that does not rely on a hold-out sample is to use a leave-one-out ML-estimator.   
In the applications and simulation studies I simply calculate (\ref{sample_mean_trafo}) on the entire sample observations for $W$.
This is theoretically not valid, yet in simulations there is practically no difference between using (\ref{sample_mean_trafo}) on the entire sample for $W$ or an equally-sized hold out sample of $W$.   

Therefore, I assume for the remainder of the paper that
\begin{equation}\label{Final Pi_dm}
\widehat{\Pi}_{dm}(x) = \int_{\mathbb{R}} \widehat \Ex[T(W, Y- \widehat m(X,W )) | X=x, W=w]d\mu(w)
\end{equation}
so when studying the asymptotics of (\ref{Final Pi_dm}) in the next Section, it is implicitly assumed that the rather slow ML-estimators dominate the asymptotic behavior of $\widehat{\Pi}_{dm}$, i.e. convergence of the sample mean to the integral is negligible compared to the convergence of ML-estimators.  


The estimator (\ref{Final_Estimator}) nests several ML-estimates and it is therefore apparent that sample splitting is required to achieve consistency.  

A particular requirement is that $\widehat m$ is calculated on a different sample than $\widehat \Pi_{dm}$ which takes  $\widehat m$ as input. 
\\
The use of sample splitting is somewhat mandatory for nested ML-estimators, see \cite{ChernPostSelec}.

The subsequent paragraph outlines the precise use of sample splitting along with a concise summary of the estimation procedure.
\paragraph{Estimation Procedure} \hfill \\
\hfill \\
\hfill
The sample is $(X_i, W_i, Y_i)$ with $i=1,\dots,n$. Set tuning parameters $K$ and $\nu$. 
\begin{enumerate}[ Step  1:]
	\item Calculate $\widehat \beta (x)$ and $\widehat{Q}$ on the full sample. 
	\item Randomly split the sample in two parts of equal size $n/2$. The two samples are referred to as sample $\mathcal{D}$ and sample $\mathcal{R}$. 
	\item Use sample $\mathcal{D}$ as training sample to learn $\widehat m$ with some ML method.
	\item Taking $\widehat m$ as given, use sample $\mathcal{R}$ to learn $ \widehat \Pi_{dm}$ with some ML method.
\end{enumerate}
Then perform cross-fitting, i.e. iterate steps 2-4 a number of $M$ times to obtain $M$ different estimates $\widehat \Pi_{dm, m}(x)$ for $m=1,\dots, M$. Then aggregate these to a final conditional RC-density estimate 
\begin{align}\label{Est::Final}
 \widehat f_{B | X}(b |x)  = \frac{1}{M}\sum_{m=1}^{M} q^K(b - \widehat\beta (x))'\widehat Q^{-1} \widehat{\Pi}_{dm, m}(x)
\end{align}
The cross-fitting procedure is optional, yet highly recommended as it stabilizes the estimates considerably. \\

In many works linking causal inference with machine learning methods, orthogonalization of treatments $W$ is often mandatory to achieve consistent estimation of causal effects, see e.g. \cite{ChernPostSelec} or at least desirable for the performance of machine learning methods, see Section 6.1.1. in \cite{grf}. Therefore this section ends with a brief discussion on how to handle the case of a demeaned $W$ in the estimation. In the next section it is shown that orthogonalization of the treatment $W$ is not innocuous in the RC model, as it slows down the convergence rate of the RC-density estimator. 
\begin{rem}\label{Orthog_W}
Additional orthogonalization of $W$ leads to a random coefficient model
\begin{align*}
Y-\Ex[Y|X,W] & = A_0 + A_1\cdot (W-E[W|X]), \\
A_0 &=  B_0 - \Ex[B_0|X]+A_1\Ex[W|X]  \;\;\text{and}\;\; A_1 = B_1 - \Ex[B_1|X]   
\end{align*}		
which does not change the interpretation of the random slope. 
Define $\Ex[W|X=x]=g(x)$ and the variable  $\overline W = W - g(X)$, then the estimation of $g$ needs to be taken into account and the following quantities reformulated to
\begin{align*}
\widetilde f_{A | X}(b - \beta(x)  |x) & = q^K(b - \beta(x))'Q^{-1}\int_{\mathbb{R}}\Ex[T(\overline W, Y- m(X,W )) | X=x, \overline W=\overline w]d\mu( \overline w) 
\end{align*}
with 
\begin{align*}
Q = \int_{\mathbb{R}^2}  (\mathcal{F}q^K )(t, t\overline{w})(\mathcal{F}q^K )'(-t, -t\overline w)d\nu(t)d\mu(\overline w).
\end{align*}  
Then the estimation procedure needs to be amended. In Step 3, the function $g$ is estimated additionally by an ML- estimator $\widehat g$. In Step 4 we use sample $\mathcal{R}$ to estimate $Q$ as well, in particular
\begin{align*}
\widehat Q = \frac{1}{|\mathcal{R}|} \sum_{i=1}^{|\mathcal{R}|}  (\mathcal{F}q^K )(t, t\cdot (W_i-\widehat{g}(X_i))) (\mathcal{F}q^K )'(-t, -t\cdot (W_i-\widehat{g}(X_i)) )d\nu(t)
\end{align*}  
\end{rem}

\section{Asymptotic Analysis}\label{Sec::Asymp}
The following section develops the asymptotic theory of the estimator in (\ref{Est::Final}) with its composite parts in (\ref{Est Q}) and (\ref{Final Pi_dm}). Estimation in the orthogonalized $W$ case outlined in Remark \ref{Orthog_W} is also considered.  
The following notation is required. 
Let $\lambda_{min}(\Omega), \lambda_{\max}(\Omega)$ denote the smallest and largest eigenvalues of a matrix $\Omega$. Further define the $L_2$-norm $\mynorm{g}=\int |g(a)|^2 da$ and the weighted $L_2$-norm  $\mynorm{g}_{\nu, \mu}=\int |g(t,x)|^2 d\nu(t)d\mu(x)$ for a generic, possibly complex-valued function $g$. To avoid confusion at some points, $\mynorm{\cdot}_E$ denotes the euclidean norm of a (complex) vector. $P_K g$ denotes the $L_2$-projection on a linear sieve space $\mathcal{B}_K$, i.e. $P_K g=\arg \min_{f\in \mathcal{B}_K} \mynorm{g - f}$. The relation $a \lesssim b$ is shorthand for $a \leq C\cdot b$ for some constant $C>0$.    

The following set of assumptions is necessary. 
\begin{assA}\label{A::Rate}
	(i) $\sup_{b\in \mathbb{R}^2} \mynorm{q^K(b)}\lesssim \sqrt{K}$ (ii) the smallest eigenvalue of $Q$ satisfies $\lambda_{\min}(Q)= O(\tau_K)$ with $\tau_K\geq 0$ and $\tau_K$ decreasing to zero and $\lambda_{\max}(\int_{\mathbb{R}^2} q^K(b)q^K(b)'db)=O(1)$ (iii) for any $x$ in the support of $X$ we have that $\mynorm{P_K{f}_{A|X=x} - {f}_{A|X=x}}=O(K^{-\alpha})$ for some $\alpha >0$ and $\mynorm{\mathcal{F}{f}_{A|X=x} - \mathcal{F}P_K{f}_{A|X=x} }_{v,\mu}=O(\tau_K \mynorm{P_K{f}_{A|X=x} - {f}_{A|X=x}})$ 
	(iv) 
	$\int_{\mathbb{R}} t^2 d\nu(t) < \infty$
	 and $\sup_{x \in \mathcal{X}} |\Pi_k(x) | \leq C_1$ for any $k$ and some $C_1 > 0$
	(v) for any $x$ in the support of $X$ it holds  $\int_{\mathbb{R}^2} \mynorm{\nabla f_{A|X}(a|x)da} \leq C_2$ for some constant $C_2>0$. 
\end{assA}
Assumption (\ref{A::Rate}) (i) is satisfied for the most commonly employed sieve bases such as splines, fourier series or wavelets, see e.g.\cite{belloni2012} as well as for Hermite functions. Sufficient conditions for Assumptions (ii) and (iii) are given in \cite{Breunig_VRC} in the absence of the measure $\mu$ in (\ref{L2cond}). With the additional measure the eigenvalue decay $\tau_K$ will generally depend on $\mu$, i.e. in my preferred specification on the distribution of $W$. Simulations show that the decay is faster the more light-tailed the distribution of $W$ and the smaller its support. Condition (iii) is a typical assumption on the approximating properties of the basis functions and the parameter $\alpha$ is solely related to the smoothness of the density functions $f_{A|X=x}$ as the dimension of the random coefficient vector is not of interest in this analysis, see e.g. \cite{Chen07} for a review of approximation properties of various sieve bases across different smoothness classes. The remaining parts (iv) and (v) are standard regularity conditions on the density $f_{A|X=x}$, in particular (iv) imposes that for any $x$ the $L_2$-projection of $f_{A|X=x}$ has bounded coefficients. 
\begin{assA}\label{A::MLRates}
(i) For any $k=1,\dots, K$ and fixed $x,w$ assume that    
\begin{align*}
& \max \Big \{ (\widehat{g}(x)-g(x))^2,  (\widehat{m}(x,w)-m(x,w))^2,  (\widehat{\beta}_0(x)-\beta_0(x))^2, (\widehat{\beta}_1(x)-\beta_1(x))^2,  \\ 
& \quad\quad\;\; (\widehat{\Pi}_{dm,k}(x)-\Pi_{dm, k}(x))^2  \Big \} = O_p\left(n^{-2\varphi} \right)
\end{align*}  
(ii) $K\tau_K^{-1}\log(K)=o(n)$ (iii) $K\tau_K\log(K)=o(n^{1-2\varphi})$ 
\end{assA}
Assumption $\ref{A::MLRates}$ (i) states an abstract upper bound for the pointwise convergence rates of various ML-estimates. Thereby one can abstract from considering different convergence rates for each ML-estimator by simply focusing on the slowest rate among those ML-estimators employed. Typically for any ML-method $\varphi< 1/2$. Part (ii) is a common rate restriction in the series estimation literature to achieve that $\mynorm{\widehat{Q}^{-1}-Q^{-1}} \to 0 $, see \cite{BellChern_Series}. The last part (iii) is an additional rate restriction that is trivially satisfied if $\tau_K^{-1}= K^\gamma$ with $\gamma >1$, i.e. if the eigenvalue decay $\tau_K$ is sufficiently fast. It holds more generally if $\varphi$ is sufficiently small.

%
  
In general the particular convergence rates depend on factors such as the effective dimension and the smoothness of the conditional expectation function that is to be estimated. An additional aspect is the proper choice of tuning parameters for any ML-method that is applied and the precise notion of high-dimensionality, i.e. the rates at which $\dim(X)$ may go to infinity relative to the sample size. In order to derive these convergence rates additional restrictions will be required. 
 
Thus Assumption $\ref{A::MLRates}$ abstracts from many theoretical and practical details of the ML-techniques employed. However, the complexity of any given ML-method makes the joint parameter choice of our model tuning parameters $K$ and $\sigma_t$ along with other parameters of the ML-routines impractical. Therefore I suppose that any ML-estimator used has been properly tuned by e.g. data-driven methods to the prediction task at hand. Thus the rate in Assumption \ref{A::MLRates} can be viewed as the best rate achievable given a set of ML-estimators that are properly tuned to their respective estimation problem. This abstraction is in line with other works using generic machine learning techniques such as \cite{ChernPostSelec} or \cite{ChernGenericML}.   
 
Below Assumption \ref{A::MLRates} is discussed in the context of regression forests which will be used in the applied segments of this paper.  
 \begin{rem}\label{Rem::ForestRates}
 Suppose estimates for the various functions summarized in Assumption \ref{A::MLRates} are obtained from applying random forest algorithms. Originally devised by \cite{Breiman01}, random forests are a popular ML-tool among practitioners and several works have since considered the asymptotic properties of random forests. Some theoretical properties like consistency have been established for various tree-growing schemes, see e.g. \cite{Biau12}, \cite{scornet2015} and \cite{WagerWalther} but the development of theory is ongoing.
 
 For obtaining pointwise results as in Assumption \ref{A::MLRates} (i) we can invoke Theorem 3.1 in \cite{AtheyWager18}. From that it follows for any estimate of a conditional expectation function that    
 \begin{align*}
n^\varphi= n^{1-b} \cdot \log(n^{b})^d 
 \end{align*}  
 where $b$ satisfies
 \begin{align*}
 b_{\min}:= 1- \left(1 + \frac{d}{\pi} \frac{\log(\omega)^{-1}}{\log((1-\omega)^{-1})}     \right)^{-1} < b < 1
 \end{align*}
 and where $\omega, \pi$ are hyperparameters of the forest algorithm, i.e. the regularity parameter and splitting probability. For additional assumptions to obtain this result, see Theorem 3.1 of \cite{AtheyWager18}. Smoothness assumptions on the underlying conditional expectation function and other regularity conditions are required. The result provides a worst-case convergence rate and obtaining optimal rates in high-dimensional settings remains an open question. The generalized random forest algorithm of \cite{grf} yields a similar result.          
 
 An additional example that will be referred to later is \cite{WagerWalther} which derive $L_2$-rates under additional sparsity assumptions, for a different class of random forest algorithms. 
 In their Theorem 4, \cite{WagerWalther} establish that for any estimate $\widehat{\tau}(x)$ of a conditional expectation function $\tau(x)$ it holds that
 \begin{align*}
 \Ex[(\widehat{\tau}(X) -\tau(X))^2 ] =O( n^{\log(\xi)/\log(2\xi)})
 \end{align*} 
 where $\xi = 1/(1-3/(4q))$ and $q$ is the effective dimension of the true conditional expectation function. See their Theorem for more details. They admit a high-dimensional setting where the number of covariates may grow with the sample size\footnote{More precisely $\lim \inf d/n > 0$} but need additional restrictions 
 on minimum effect sizes of some covariates, see in particular their Assumptions 3 and 4.   
 \end{rem}
The Remark above gives convergence rates for random forest estimators of conditional expectation functions. Further Assumption \ref{A::MLRates} imposes a convergence rate on $\widehat{\Pi}_{dm,k}(x)$ which is by definition (\ref{Final Pi_dm}) itself an average of different random forest estimates by averaging over $W$. Thus its convergence rate can be expected to be faster then the rate of the random forest estimate $\widehat \Ex[T(W, Y- \widehat m(X,W )) | X=x, W=w]$  
for any fixed $w$. \\

Finally the following convergence rate result holds.      
\begin{theo}\label{Thm::Rate} 
	Under Assumptions \ref{A::Ident1}- \ref{A::MLRates} it holds that
	\begin{align*}
	\int \left[\widehat{f}_{B|X}(b|x) - {f}_{B|X}(b|x)\right]^2db = O_p(\tau^{-2}_K \frac{K}{n^{2\varphi}} + K^{-\alpha})
	\end{align*}
\end{theo}
Here the decay of eigenvalues $\tau_K$ of the matrix $Q$ serves as a measure of ill-posedness. The estimation of the RC density is known to be an ill-posed problem which implies a slower convergence rate of estimators, see \cite{hodermam_RC_2010} or \cite{Breunig_VRC}. If $\tau_K$ decays polynomially, e.g. $\tau_K \sim K^{- \gamma/2}$ then choosing $K\sim n^{\frac{2\varphi}{1+\gamma+\alpha}}$ balances bias and variance and results in the convergence rate 
\begin{align*}
\int \left[\widehat{f}_{B|X}(b|x) - {f}_{B|X}(b|x)\right]^2db = O_p(n^{-\frac{\alpha}{1+\alpha+\gamma}2\varphi}).
\end{align*}

This shows that the convergence rate $\varphi$ of the generic machine learning estimators is slowed down by a factor $\alpha/(1+\alpha+\gamma) < 1$. This loss of speed is increasing in the eigenvalue decay parameter $\gamma$ and decreasing in the smoothness parameter $\alpha$. Note that the density considered  here is bivariate and thus there is no explicit parameter for the number of random coefficients. If $W$ is multidimensional, its dimension enters the factor and further slows down convergence.  

This rate result is not sharp in that the convergence rate can be improved for a given ML-technique. As $\varphi$ depends on tuning parameters specific to the chosen ML- technique, a joint choice of $K$ and the tuning parameters subsumed in $\varphi$ may improve the rate of convergence. However calculating these exact rates may prove difficult and in practice, joint tuning of $K$ along with the parameters of the specific ML-techniques leads to excessive computational costs. Further note that using the result from \cite{WagerWalther} in Remark \ref{Rem::ForestRates} a rate for $\Ex[\int \left[\widehat{f}_{B|X}(b|X) - {f}_{B|X}(b|X)\right]^2db]$ can be derived analogously. \\

In Remark \ref{Orthog_W}, estimation with orthogonalized treatment $W$ is considered. This is important for some ML estimators applied in causal inference. For the estimation of $\beta_1(x)=\Ex[B_1|X=x]$, \cite{grf} suggest orthogonalization of both the outcome $Y$ and treatment $W$ to improve the performance of the random forest routines involved, see Section 6.1.1 in \cite{grf}. The subsequent Corollary shows that in the context of random coefficient models orthogonalization of $W$ is not innocuous, as it slows down the convergence rate of the RC-density estimator. This is in contrast to Theorem \ref{Thm::Rate} where sole orthogonalization of the outcome $Y$ does not result in a slower rate.  
\begin{coro}
	Let Assumptions \ref{A::Ident1}- \ref{A::MLRates} hold. Consider the orthogonalized W case outlined in Remark \ref{Orthog_W}. 
	Assume additionally that $K=o(\tau_K^{-1})$ then
	\begin{align*}
		\int \left[\widehat{f}_{B|X}(b|x) - {f}_{B|X}(b|x)\right]^2db = O_p(\tau^{-2}_K \frac{K^2}{n^{2\varphi}} + K^{-\alpha})
	\end{align*}
\end{coro}
This slower convergence rate is due to the fact, that the "generated regressor" $W-\widehat{g}(X)$ appears within Hermite functions $q^K$. As is shown in the proof, the derivatives of Hermite functions $q^K$ diverge in $K$ and thus an additional $K$-term appears in the derivation of the convergence rate. This is a general issue, that does not seem to have been noticed so far. The convergence rate of any Hermite function sieve estimator of a RC density is slower when $W$ is a generated regressor.  
\section{Inference}\label{Sec::Inference}

This section discusses inference for the conditional RC-density estimator, in particular pointwise inference of the conditional RC density function estimate. Asymptotic normality of the RC-density estimator follows from asymptotic normality of ML-estimators linked with theory from the series estimation literature. For random forests such asymptotic normality results have been recently provided by \cite{grf} and \cite{AtheyWager18}. 
The main issue is how to establish the asymptotic covariance of the $K$ different ML-estimators, i.e. to obtain estimates for $\Ex[\widehat{\Pi}_{dm,j}(x)\cdot\widehat{\Pi}_{dm,l}(x)]$ for any pair $1\leq j,l \leq K$. 
This issue can be overcome by introducing an additional layer of sample splitting. If we split the sample $\mathcal{R}$ on which sieve coefficients are learned into $K$ equally sized subsamples and use a different subsample for estimating each sieve coefficient then naturally these estimators are stochastically independent. Then the asymptotic variance of each sieve coefficient can be obtained by resigning to established variance estimators for the respective ML-method. 
This approach requires the block size  $|\mathcal{R}|/K $ to be of meaningful size in practice, yet again cross-fitting, i.e. iterated  sample splitting and averaging of estimates stabilizes the results and reduces any losses in efficiency. \\

More precisely, moving forward I amend Step 4. of the estimation procedure at the end of Section \ref{Sec::EstIdea}.
\begin{enumerate}[ Step  4:]
	\item Additionally split $\mathcal{R}$ into $K$ subsamples $\mathcal{R}_1, \dots, \mathcal{R}_K$  of size $\lfloor|\mathcal{R}|/K \rfloor$ and calculate each $\widehat{\Pi}_{dm,j}$ using subsample $\mathcal{R}_j$.
\end{enumerate}
Introduce the notation $r_p(n)=n^\varphi$.
The resulting estimator does not coincide with the one in the previous section. The convergence rate is similar with the term $r_p(n/K)$ appearing in the convergence rate rather then $r_p(n)$ which is due to the fact that each sieve coefficients is learned on a sample of size $n/K$. Nevertheless in implementations the finite sample performance is comparable to the estimator in the previous section. 
\begin{assA}\label{A::Norm_I}
(i) For each $k=1,\dots,K$ there exists a non-increasing sequence $\sigma_{dm,k}(x)$ such that
\begin{align*}
\frac{\widehat{\Pi}_{dm,k}(x)-\Pi_{dm,k}(x)}{\sigma_{dm,k}(x)}\overset{d}{\rightarrow}  N(0,1) 
\end{align*}	
where $\sigma_{dm,k}(x) \propto r_p(n/K)^{-1}$ and $\min_k \sigma_{dm,k} > 0 $

(ii) $\max_{k} \sigma_{dm,k} / \min_{k} \sigma_{dm,k} = O(1) $

(iii) Define 
\begin{align*}
\Sigma_n(x)  := & diag\left(\sigma_{dm,1}(x), \dots, \sigma_{dm,K}(x) \right) \\
v_n(b,x) := & \mynorm{q^K(b)'Q^{-1}\Sigma_n(x)}_E =  \sqrt{q^K(b)'Q^{-1}\Sigma^2_n(x)Q^{-1}q^K(b) }
\end{align*}
it additionally holds that 
\begin{align*}
\frac{q^K(b)Q^{-1}\left( \widehat{\Pi}_{dm}(x)-\Pi_{dm}(x)  \right)}{v_n(b,x)} \overset{d}{\rightarrow}  N(0,1) 
\end{align*}
with $v_n(b,x)$ bounded away from zero
(iv)  $\sqrt{K/r_p(n/K) }\tau_K^{-1}=o(1)$.   
\end{assA}
Assumption \ref{A::Norm_I} (i) establishes asymptotic normality of various ML-estimators. Such asymptotic normality results are standard for many ML-methods. For (honest) random forests such results have been established by \cite{AtheyWager18}. For asymptotic normality results for different tree-based algorithms see the references therein and also \cite{grf}.
$\sigma_{dm,k}$ is the individual standard error of the $k$-th ML step and subsumes both the convergence rate and the residual standard deviation. 
Part (ii) of the above assumption imposes that the residual standard deviation in any of the $K$ ML-regressions is bounded away from zero and infinity. This assumption can be weakened to allow for standard deviations to diverge as $K$ grows at the cost of introducing an additional rate parameter that further strengthens rate restrictions.
In Assumption \ref{A::Norm_I} (iii) $v_n(b,x)$ serves as the standard error of the conditional RC-density estimate. It holds that 
\begin{align*}
v_n(b,x) \gtrsim \sqrt{K}\tau_K^{-1} \min_{k} \sigma_{dm,k}
\end{align*}
which is bounded away from zero and approaching zero asymptotically under the rate restriction in (iv) which is required for consistency of the RC-density estimate, see Theorem \ref{Thm::Rate}.

Part (i) and (iii) are the most intricate conditions in Assumption \ref{A::Norm_I} and typically involve additional regularity conditions and restrictions on the growth of $K$. The following Lemma gives conditions such that Assumption \ref{A::Norm_I} (i) and (iii) are satisfied for honest regression forests. 
\begin{lem}\label{Lem::ForestAN}
 Assume the following conditions hold:
 (i) The density $f_X$ is bounded away from zero and infinity 
 (ii) for any $k$ in $1,\dots, K$ the function $\Pi_{dm, k}(x)$ is Lipschitz continuous and also $\Ex[T_k(W, Y)^2|X=x]$ is Lipschitz continuous.    
 (iii) for any $k$ and uniformly in $x$ it holds $Var(T_k(W, Y)|X=x) > 0$ and $\Ex[ | T_k(W,Y) - \Ex[T_k(W,Y)|X=x] |^{2+\delta} |X=x]< M $ for some constants $\delta, M > 0$. 
 (iv) $K=o(r_p(n/K)^c)$ with $c=\min\{\delta/2, 1, -\beta^{\ast}/b \}$ and $\beta^{\ast}=1+\epsilon-b/\beta_{\min}<1 $. 
 Then if $\widehat{\Pi}_{dm,k}(x)$ is an honest random forest estimator in the sense of Theorem 3.1. of \cite{AtheyWager18} then Assumption \ref{A::Norm_I} (i) is satisfied under conditions (i)-(iii). If additionally condition (iv) and Assumption \ref{A::Norm_I} (ii) are satisfied then Assumption \ref{A::Norm_I} (iii) holds  . 
\end{lem}
Assumptions (i) to (iii) in the Lemma above are as in Theorem 3.1. of \cite{AtheyWager18} that establishes asymptotic normality of a single (honest) random forest estimator of a mean regression function. Assumption (iv) is an additional rate restriction that is needed to achieve asymptotic normality of the sieve coefficient estimates. Additional rate restrictions are common in the series estimation literature, see e.g. Theorem 4.2. (iii) in \cite{BellChern_Series} and also appear in Assumption 5 (ii) of \cite{Breunig_VRC} for an RC-density estimate. Here the rate restriction is milder then the one in \ref{A::Norm_I} (ii) if for instance $\delta>2$ and the convergence rate $b$ is sufficiently fast such that $-\beta^{\ast}/\beta_{\min} > 1$.
In this case the rate restriction in Assumption \ref{A::Norm_I} (iv) that guarantees consistency of the RC-density estimate is already sufficient for Assumption \ref{A::Norm_I} (iii). 
If however $\delta$ is rather small and the convergence rate $b$ close to the worst-case $\beta_{\min}$, there can be cases where the rate restriction of Lemma \ref{Lem::ForestAN} is stronger compared to the one in \ref{A::Norm_I} (ii), especially if the decay of $\tau_K$ is slow. 

The following additional assumption is required. 
\begin{assA}\label{A::Norm_II}
 For any $x$ in the support of $X$ and for any $a\in \mathbb{R}$ it holds $P_K f_{A|X}(a|x) - f_{A|X}(a|x) =o(v_n(b,x)) $. 
\end{assA}
Assumption \ref{A::Norm_II} is an undersmoothing condition that is standard for pointwise inference of a series estimator, see \cite{BellChern_Series} (4.18). Note that similar rate restrictions are not needed for $\widehat{\beta}(x), \widehat{g}, \widehat{m} $, as these are calculated on a sample proportional to $n$ and thus converge at rate $r_p(n)$ which is always faster then the standard error rate $v_n(b,w)$.
An estimator for $v_n(b,x)$ is
\begin{align*}
\widehat{v}_n(b,x)=\mynorm{q^K(b- \widehat{\beta}(x))\widehat{Q}^{-1}\widehat{\Sigma}_n(x) }_E
\end{align*}
where $\widehat{\Sigma}_n(x)  = diag\left(\widehat\sigma_{dm,1}(x), \dots, \widehat\sigma_{dm,K}(x) \right)$ and the individual standard error estimates $\widehat\sigma_{dm,k}(x)$ are specific to the employed ML-method. For random forests these can be obtained from applying the infinitisimal jackknife procedure of \cite{Efron2014}, see also the discussion in \cite{AtheyWager18}. Additional rate restrictions required for consistent estimation of the standard error $v_n(b,x)$ are not required. In the proof of the subsequent Theorem \ref{thm::Inference} it is shown that $\widehat{v}_n(b,x)$ is consistent for $v_n(b,x)$ under the assumptions given so far.   

The following pointwise asymptotic normality result holds
\begin{theo}\label{thm::Inference}
	If Assumptions \ref{A::Ident1}-\ref{A::Norm_II} are satisfied then,
	\begin{align*}
	\frac{\widehat{f}_{B|X}(b|x) - f_{B|X}(b|x)}{v_n(b,x)} \overset{d}{\rightarrow} N(0,1)
	\end{align*}
	and further
	\begin{align*}
		\frac{\widehat{f}_{B|X}(b|x) - f_{B|X}(b|x)}{\widehat v_n(b,x)} \overset{d}{\rightarrow} N(0,1).
	\end{align*}
\end{theo}
This determines the asymptotic normality of the estimator conditional on a given sample split, i.e. the case $M=1$. To handle the cross-fitting case $M>1$ and additional uncertainty due to sample splitting we can follow the variational inference approach of \cite{ChernGenericML}.
The idea summarizes as follows. 

Suppose there are $M$ different estimates $\widehat{f}_{B|X}^l(b|x)$ for $l=1,\dots, M$. For each estimate it is possible to construct a (1-$\alpha$)-confidence interval $[L_{1-\alpha, l}, U_{1-\alpha, l}]$ from Theorem \ref{thm::Inference} with $L_l=\widehat{f}_{B|X}^l(b|x)-c_{1-\alpha} \cdot \widehat{v}_n(b,w)$ and $U_l=\widehat{f}_{B|X}^l(b|x)+c_{1-\alpha} \cdot \widehat{v}_n(b,w)$ and $c_{1-\alpha}$ denoting the respective $1-\alpha$ quantile of the standard normal distribution.    

To construct an asymptotically valid $1-\alpha$- confidence intervals for $f_{B|X}(b|x)$ \cite{ChernGenericML} propose $[ \underline{Med}(\{L_{1-\alpha/2, l}\}_{l=1}^{M}), \overline{Med}(\{U_{1-\alpha/2, l}\}_{l=1}^{M})]$ with $\underline{Med}$ denoting the lower median and $\overline{Med}$ the upper median. The confidence level of each single interval needs to be discounted to $1-\alpha/2$. \cite{ChernGenericML} provide a similar reasoning for constructing adjusted p-values. 

\section{Marginal Densities, Variable Importance Measures and Cross-Validation}\label{Sec::CVetc}
This section touches on additional important aspects for the practical application of the estimation procedure. First, I discuss how to construct estimates of the marginal random coefficient density $f_{B}$. Second, I present a measure of variable importance that assigns an importance score to every variable in $X$. This is an important descriptive tool for uncovering which variables in $X$ drive the heterogeneity in conditional RC densities.
Lastly, I discuss a cross-validation procedure for a data-driven choice of tuning parameters.  
\paragraph{Estimating marginal RC densities}
There are various direct estimators for marginal RC densities in the literature such as the Radon transform estimator of \cite{hodermam_RC_2010} or an adaptation of the sieve estimation strategy from \cite{BreunigHoder18} and \cite{Breunig_VRC}. The common identifying restriction is however full independence between $B$ and $W$ which is difficult to maintain in non-experimental data settings. 

Maintaining the weaker conditional independence condition in Assumption \ref{A::Ident1} (i) estimates of the marginal density can be readily constructed by averaging over leave-one-out estimates of conditional density estimates $\widehat{f}_{-i, B|X}$. Here the estimate is calculated without using the $i$-th datapoint $(Y_i, W_i. X_i)$. A consistent estimator for the marginal density is 
\begin{align*}
\widehat{f}_B(b) = \frac{1}{n} \sum_{i=1}^{n}\widehat{f}_{-i, B|X}(b|X_i)
\end{align*}

The estimator $\widehat{f}_B$ will inherit its asymptotic properties from the conditional estimate $\widehat{f}_{B|X}$ which is discussed in the previous section. Thus the convergence rate is slower compared to direct marginal RC density estimators making use of full independence between random coefficients and covariates. To the best of my knowledge there are however currently no alternative estimators for the marginal random coefficient density that operate under Assumption \ref{A::Ident1} (i).  
\paragraph{Variable Importance Measures}
The estimation procedure outlined so far yields consistent estimates of $f_{B|X=x}$ for any given point $x$. An important question in applications is to identify those variables in the set of controls $X$ that drive heterogeneity in conditional RC densities, i.e. a criterion to guide the choice of interesting points $x$ on which to evaluate the estimate $\widehat f_{B | X}(b |x)$.

We focus here on our running example that makes use of regression forests. 
Note that for regression forests generally no post-selection inference problems arise as variable selection is done within in the various ML-steps of the estimation procedure. 
The goal is to find points $x$ that reveal interesting heterogeneities to the researcher. This is analogous to the role of variable importance measures for the causal forests of \cite{grf}. 

For each of the ML-estimators used we can calculate a measure of variable importance that assigns an importance score to each covariate that is normalized to sum to one.
This score is informative on how often a specific variable has been used for placing splits in the growing of the forest. 
  
First, I focus on the ML-estimates $\widehat{\Pi}_{dm}$ which constitute the sieve coefficients and thus determine the shape of the RC density. For each $k=1,\dots, K$ let $VI_k(X_l)$ denote an importance score assigned to covariate $X_l \in X$ by the regression forest estimator $\widehat{\Pi}_{dm}$.
Here any normalized measure such as, cite can be used.  

To obtain a global measure of variable importance for the shape of the function $f_{A|X=x}$ we can simply average over $K$. 
Thus define the variable importance of $X_l$ for the shape of the density as
\begin{align*}
VI_{shape}(X_l)= \frac{1}{K}\sum_{k=1}^{K} VI_k(X_l)
\end{align*}
A measure of variable importance for the conditional expectation of random coefficients  $\beta(x)$ is directly available by considering the variable importance measure for causal forests as implemented in the \cite{grf}-package. In contrast to $VI_{shape}$, this measure of variable importance for the center of the density will be henceforth referred to as $VI_{mean}$.  

\paragraph{Parameter Tuning}
In this paragraph I propose a cross-validation procedure for the choice of tuning parameters. 
Analogous to classical density estimation tuning parameters are chosen by minimization of the integrated squared error 
\begin{align*}
 \argmin_{K, \sigma_t} ISE(K, \sigma_t) := \int_{\mathbb{R}^2} \left( \widehat{f}_{B|X}(b|x, K, \sigma_t) - f_{B|X}(b|x)\right)^2 db 
\end{align*}
which is equivalent to minimizing the criterion
\begin{align*}
 J(K, \sigma_t) := \int_{\mathbb{R}^2} \widehat{f}_{B|X}(b|x, K, \sigma_t)^2 db - 2\int_{\mathbb{R}^2} \widehat{f}_{B|X}(b|x, K, \sigma_t)f_{B|X}(b|x)db
\end{align*}
The first part is simply the integrated squared RC density estimate. The second term is typically estimated via cross-validation. However it is not possible to observe realizations of random coefficients. 
The following Lemma links the second part to an expression that can be estimated via leave-one-out cross validation.   
\begin{lem}\label{Lem::CV}
	Let Assumption \ref{A::Ident1} hold, then the following identity holds
	\begin{align*}
	\int_{\mathbb{R}^2} \widehat{f}_{B|X}(b|x)f_{B|X}(b|x)db = \int_{\mathbb{R}} \Ex[V(Y,W)'\widehat Q^{-1}\widehat{\Pi}_{dm}(X) |X=x, W=w]dw
	\end{align*}
	where $V(y,w)=(V_1(y,w), \dots, V_K(y,w))$ with
	\begin{align*}
	V_k(y,w)=\frac{1}{2\pi^2} \int q_k(b)\cdot|t|\cdot \exp[it(y-b'(1,w))]dtdb
	\end{align*}
\end{lem}
Here again a weighting for $w$ should be considered for practical reasons. So defining 
\begin{align*}
V_k(y,w)=\frac{1}{2\pi^2} \int q_k(b)\cdot|t|\cdot \exp[it(y-b'(1,w))]/f_W(w)dtdb
\end{align*}
it is equivalent to consider the integral $\int_{\mathbb{R}} \Ex[V(Y,W)'\widehat Q^{-1}\widehat{\Pi}_{dm}(X) |X=x, W=w]f_W(w)dw$.
Using a plug-in estimate for the unknown density $f_W$ the function $V$ can be computed and cross-validation used to estimate the conditional expectation with a machine learning estimator. Here either a subsample of observations that has not been used for calculating $\widehat{f}_{B|X}$ can be used for the prediction task or a leave-one-out estimator for $\widehat{f}_{B|X}$. Standard practices of cross-validation apply. 

\section{Monte Carlo Simulations}\label{Sec::MCs}
This section evaluates the finite sample performance of the RC-density estimator outlined in the earlier sections. The following data generating process is studied first,
\begin{align}\label{MC::DGP}
Y & = B_0+ B_1\cdot W, \;\; \text{with} \\
B_0 &= \sin(X_1)+ A_0 \nonumber \\
B_1 &= X_2+ 0.5\cdot X_3+0.25\cdot X_2 \cdot X_3 + A_1 \nonumber \\
W & =  1+ X_3 + (1+X_3^2)\cdot V \nonumber
\end{align}
where $A_0, V$ are standard normal random variables and $A_1$ is a mixture of a $N(-1.5,1)$ and a $N(1.5, \sqrt{1/2})$ random variable with weights $1/2$.
 In this setting the density of the random slope $B_1$ is bi-modal and any testpoint $X=x$ solely determines the center of the density function.    
 The controls $X$ are a $p$-dimensional vector of iid standard normal variables. Here I set $p=10$ but as we see from the setup above, only variables $X_1, X_2, X_3$ are of importance in this toy model. This reflects the common practical problem that there is a large set of control variables but only some of them drive the heterogeneity in $B_1$ or may otherwise affect the outcome $Y$. Further I introduce a form of heteroskedasticity in the equation for $W$, such that we do not only consider the clean case where orthogonalization removes all dependence between $W$ and $X$. Further there is some form of dependence between the regressor $W$ and the random slope $B_1$ as both depend on the regressor $X_3$.

The goal is to estimate the density of the random slope $B_1$ conditional on some testpoint $X=x$. Here I implement the estimator in (\ref{Est::closed_no_mu}) with the algorithm outlined at the end of Section \ref{Sec::Est}. In this setting $W$ does not need to be orthogonalized.  

The other parameters of the estimation problem are chosen as follows. I set $K_1=K_2=3$ and thus there are a total number of $K=9$ basis functions. Hermite polynomials are used as sieve basis $q^K$ and the weighting measure follows a log-normal law, i.e. $\mu \sim lognormal(0, \sigma_t)$ with $\sigma_t=1$. In practice, when only the slope parameter is of interest, $K_1$ should be fixed and cross-validation  performed to guide the choice of $K_2$ and $\sigma_t$. Simulations show that $K_2$ is the more relevant parameter for estimates compared to $\sigma_t$, so sole cross-validation of $K_2$ may be sufficient if computation time is a concern.   
To reduce computational effort parameters in this simulation study are not chosen via cross-validation and there is no cross-fitting as well. So $M=1$ and the sample is split only once in equally sized parts $\mathcal{R}$,$\mathcal{D}$ of size $n/2$ and RC-density estimates are computed only once per Monte Carlo iteration. 
The testpoint is chosen as $x=(0,0.3,0,\dots,0)$, so the correct density is centered around $0.3$. 

All ML estimates are obtained from using honest regression forests, respectively causal forests for the quantity $\beta_1(x)$, see \cite{AtheyWager18} and \cite{grf}, with the implementation taken from the $\textbf{grf}$-package in $R$. Each random forest is tuned using implemented data-driven routines, the number of trees in each forest is set to 2000, which is the packages default setting. 
In general, I find that tuning of internal forest parameters does improve the quality of estimates but is only of secondary importance for the overall shape of the density estimate.\\

The sample size is $n=1000$ and 100 Monte Carlo draws of the model in (\ref{MC::DGP}) are performed. The simulation results for $\widehat f_{B_1|X=x}$ are presented in Figure \ref{Plot::Sim_Mix_SD}.  
\begin{figure}[!htbp]
	
	\centering
	\includegraphics[height=0.33\textheight, width=0.75\textwidth]{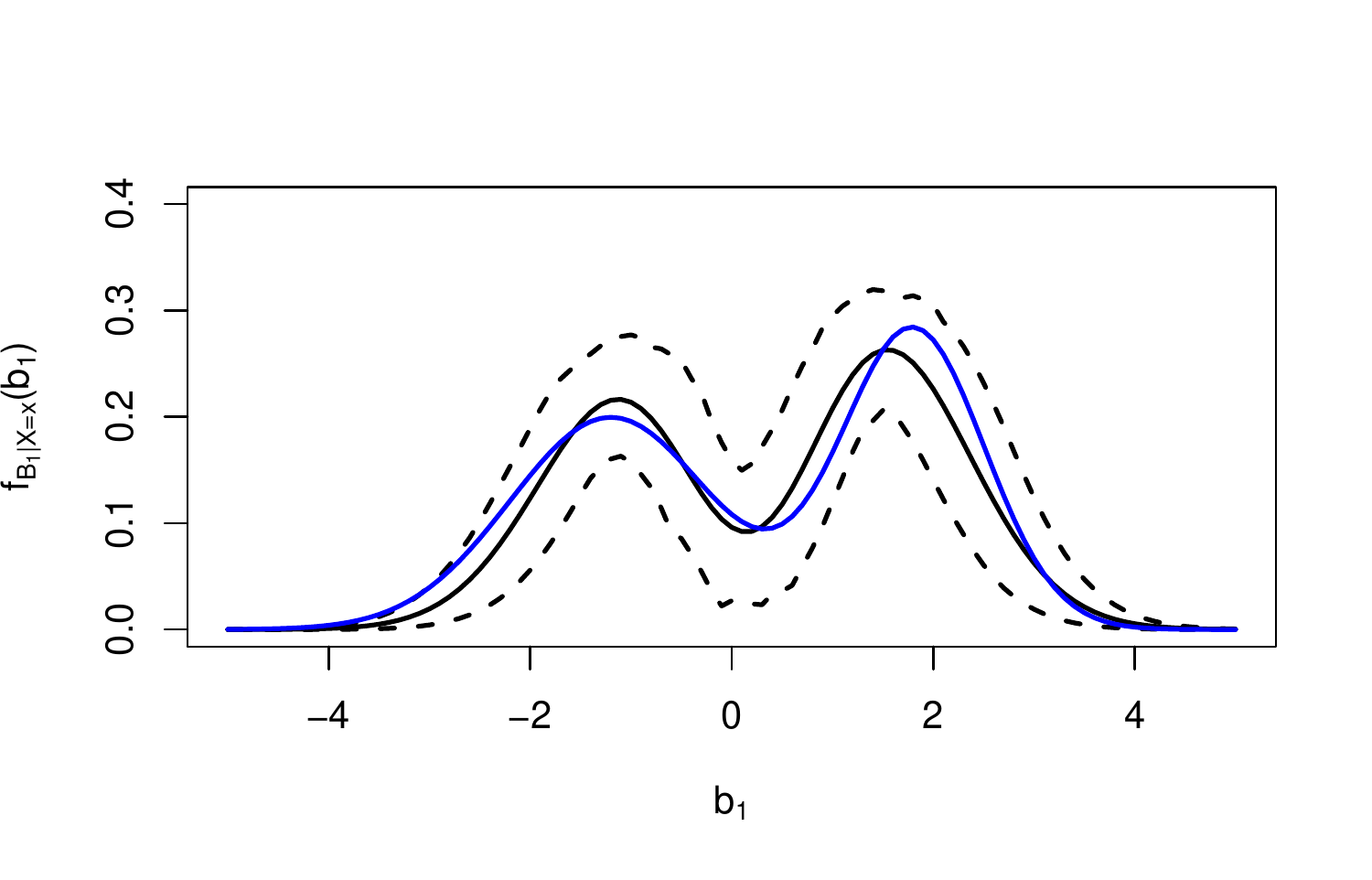}
	
	\caption{The solid black line denotes the median of the Monte Carlo estimates. The dotted lines the $95\%$- and $5\%$-quantiles. The solid blue line is the correct density. Key parameters: $K_2=3$ and $\sigma_t=1$ }
	\label{Plot::Sim_Mix_SD}	 
\end{figure} 
%
%
Figure \ref{Plot::Sim_Mix_SD} shows a favorable performance of the estimator even for a moderate sample size and for a coarse choice of $K_2$. \\ 

The second data generating process is,
\begin{align}\label{MC::DGP2}
Y & = B_0+ B_1 \cdot W, \;\; \text{with} \\
B_0 &= \sin(X_1)+ A_0 \nonumber \\
W & =  1+ X_3 + V\cdot (1+X_3^2) \nonumber
\end{align}
where all random variables are chosen as before and $B_1$ is a mixture distribution like $A_1$ in the first setting, but now with weights $\Phi(X_2), 1-\Phi(X_2)$.
So in this setting $X$ determines the entire shape of the density function as opposed to the first setting where $X$ only determines the center of the density. For the testpoint $x=(0,0.3,0,\dots,0)$, the conditional density is again bi-modal but now the mode on the negative part of the domain is more pronounced.   
All parameters and hyperparameters of the ML-procedures are as before but now $K_2=7$ to illustrate the performance of the estimator for a more complex model. As the density of $B_1$ is more dispersed compared to the first setting this increase in complexity can be rationalized. Note that the support of each of the Hermite basis functions increases with $K$. This leads to the suggestion to increase $K$ for highly dispersed densities or to otherwise scale down $Y$ and $W$ accordingly to control the maximal dispersion of the density.  
The simulation results are presented in Figure \ref{Plot::Sim_Class}.
\begin{figure}[!htbp]
	
	\centering
		\includegraphics[height=0.33\textheight, width=0.75\textwidth]{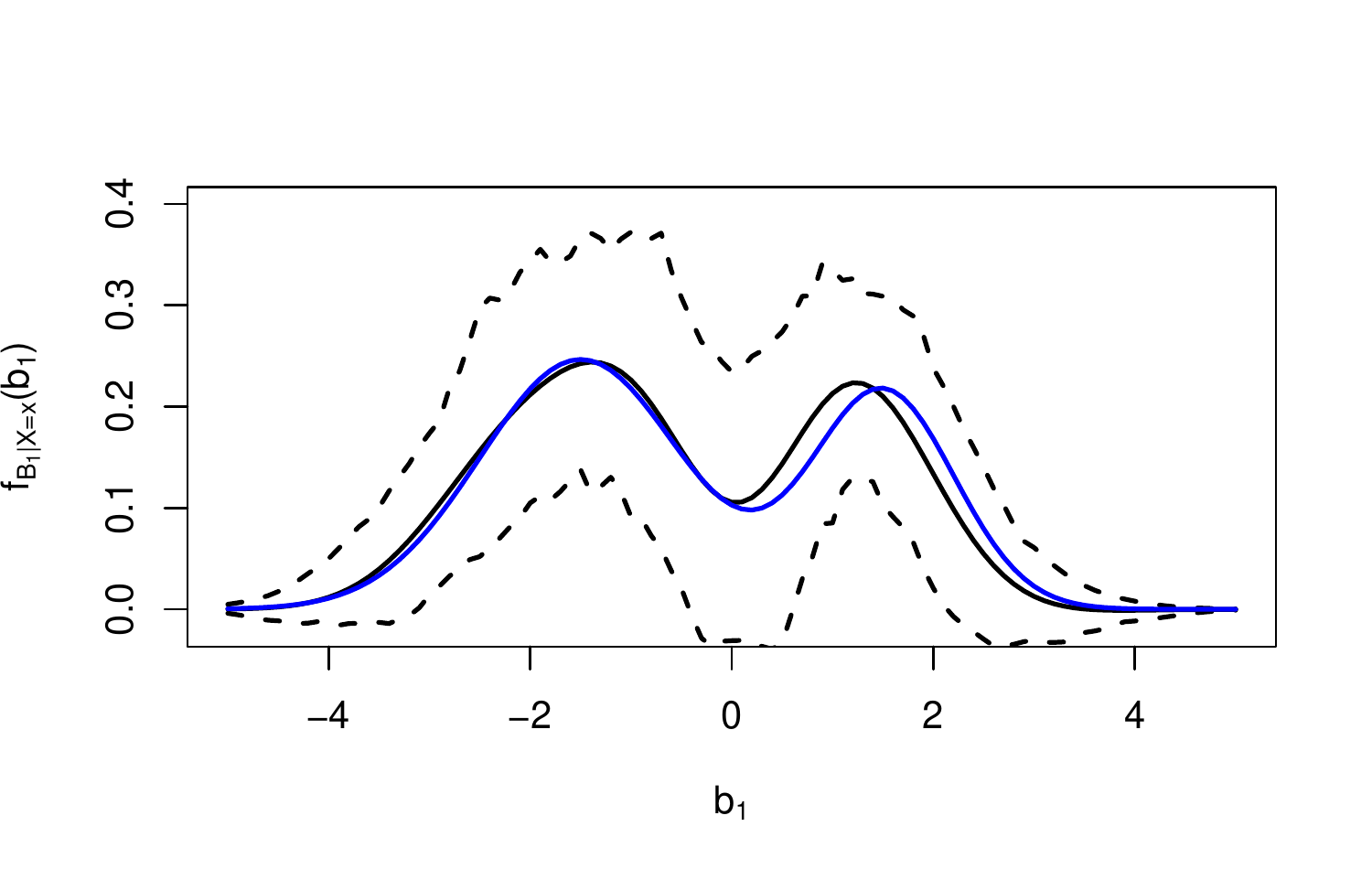}
	\caption{
		The solid black line denotes the median of the Monte Carlo estimates. The dotted lines the $95\%$- and $5\%$-quantiles. The solid blue line is the correct density. Key parameters: $K_2=7$ and $\sigma_t=1$. }
	\label{Plot::Sim_Class}	 
\end{figure}
Through the larger number of basis functions the bias is comparably lower then in Figure \ref{Plot::MargDensity} at the expense of increased confidence intervals. As there are no shape constraints, we see that density estimates can in principal have negative parts. 
Yet, the method can detect the conditional density reliably even when the entire shape of the conditional density varies with $X$. 


\section{Empirical Application}\label{Sec::Application}
    In this section the estimation strategy is applied to study heterogeneous effects of stock market expectations on portfolio choice. 
    I make use of the innovation sample of the german socio-economic panel (SOEP-IS). Therein survey respondents were supplied with a hypothetical amount of $50,000$ Euros and asked to split their investment among one risk-free and one risky asset with returns paid out one year later. 
    The risk-free asset is a state claim with a fixed annual interest rate of $4\%$ whereas the risky asset's return hinges on the return of the german stock market index (DAX) within the next year. 
    
    This experiment has been previously analyzed in \cite{HuckWeiz15}. Economic theory suggests that stock market expectations and  risk preferences are the main determinants of the portfolio choice task at hand. 
    
    The goal is to study the effect of stock market expectations on the investment in the risky asset. 
    Formulated as a random coefficient model I study the following econometric model,
    \begin{align*}
    Y_i = B_{0,i}+ B_{1,i}\cdot W_i
    \end{align*}
    where $Y_i$ denotes the individual investment in the risky asset, $W_i$ is the individual belief on the development of the DAX for the next year and $B_{1,i}$ is the individual effect of interest. The random intercept $B_{0,i}$ subsumes the effects of other controls $X$ and further unobservable characteristics on the outcome $Y_i$. 
    \begin{table}[h]
    	\centering
    	\begin{tabular}{l r r r r r r r r}
    		& Min. & 1. Quant & Median & Mean & 3. Quant. & Max.   \\
    		\hline\\[-1.0em]
    		$Y$ (in Euro) &  1000 &  15000 &  25000 &  24029 &  30000  & 50000    \\
    		$W$ (in $\%$-points)& -50 &  2 &   5 &  4.90 &   8  &130 \\
    	\end{tabular} 
    	\caption{Summary Statistics}\label{SummaryStats}
    \end{table}
    The set of controls is quite rich and contains 75 variables including information on socio-demographics such as gender, age or tertiary degrees as well as self-assessed measures of risk aversion, personality traits or skills in mathematical calculations. Summary statistics for the main variables are provided in Table \ref{SummaryStats}. 
    
    In order to apply the estimation method we need to assume that the conditional independence restriction of Assumption \ref{A::Ident1} (i) holds. Applied to the present setting this implies that stock market beliefs are exogenous conditional on the set of controls $X$. 
    The data is observational and beliefs are self-reported, so we cannot rule out relations between beliefs and other controls which rules out considering the standard, unconditional RC model that relies on full independence of random coefficients and controls. 
    
    Further beliefs $W$ must vary sufficiently in the population to plausibly fulfill the support restrictions in Assumption \ref{A::Ident1} (ii), which is the case in this setting. \\
    
    The analysis begins by choosing the tuning parameters $K$ and $\sigma_t$. 
    As the random slope is of main interest, I fix $K_1=3$ and $\sigma_t=1$ and
%
%
     vary the choice of $K_2$. Figure \ref{Plot::Eyeballing} presents various estimates for different choices of $K_2$ and a suitable choice of $K_2$ can be eyeballed. For most choices of $K_2$ the two modes of the density are centered as for the case $K_2=5$ which thus appears to be a reasonable and coarse choice for the remainder of this analysis.    
     
    Next, I set a testpoint $x$ that corresponds to the medians of the variables in $X$. For those individuals with "median" characteristics $X=x$  an estimate of the random slope density $f_{B_1|X=x}$ is presented in \ref{Plot::MargDensity} below. 
    \begin{figure}[!htbp]
    	\centering
    	\includegraphics[height=0.33\textheight, width=.75\textwidth]{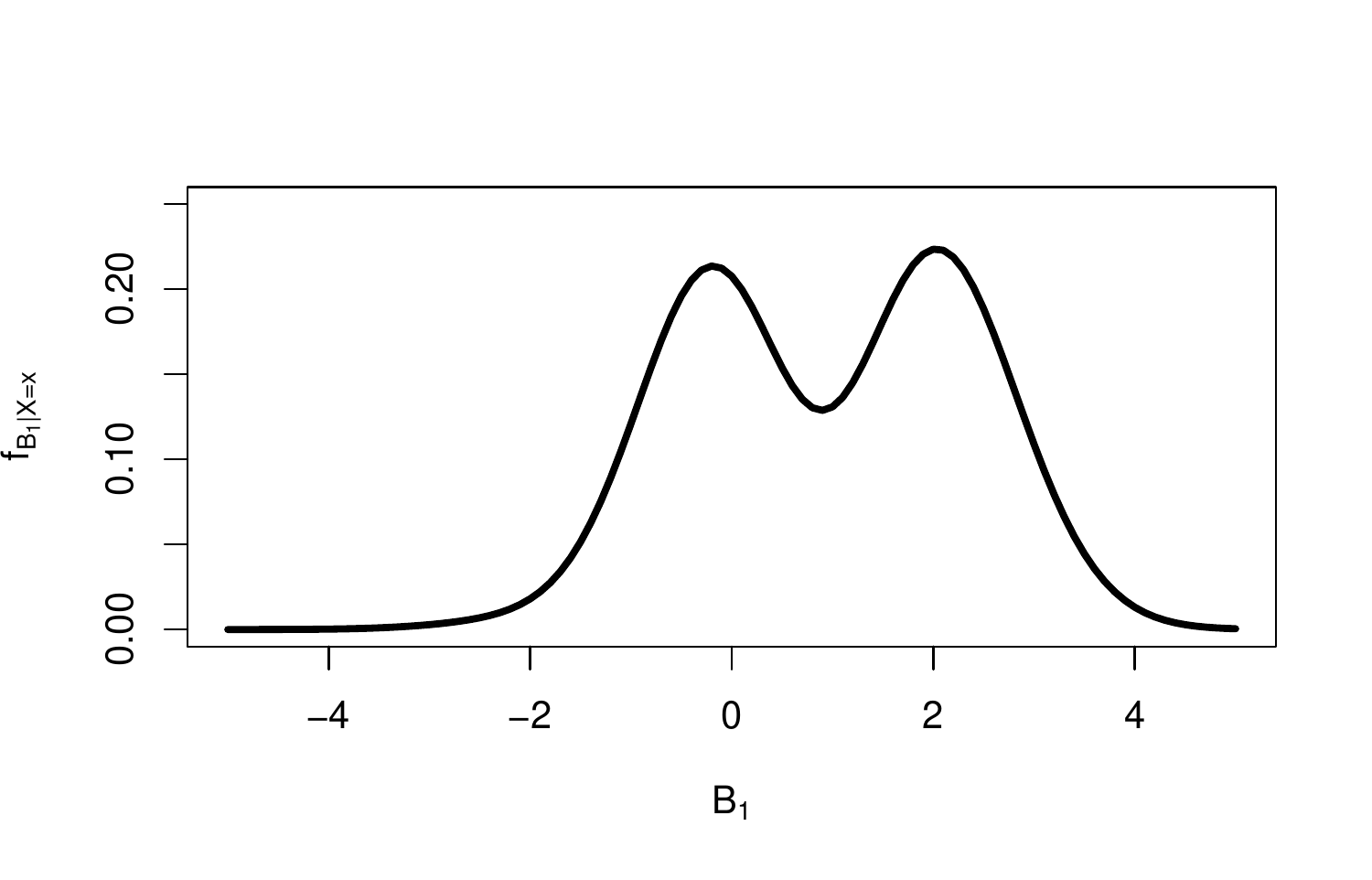}
    	
    	\caption{Estimate of the conditional density $f_{B_1|X=x}$. The tuning parameters have been chosen as $K_1=3$, $K_2=5$ (in total $K=15$) and $\sigma_t =1$. $M=100$ sample splits are performed. The testpoint $x$ is chosen as the medians of the variables in $X$.} 
    	\label{Plot::MargDensity}	 
    \end{figure}
Most notable is the bi-modal shape of the density with one mode centered around zero and another around 2. Note that variables $Y$ and $W$ have been rescaled such that a value of 2 can be interpreted in the following way. A $1\%$-point increase in beliefs is associated with investing 1000 Euro (that is $2\%$ of available funds) more into the risky asset. 

Such a bi-modal density corresponds to the existence of two types in the population. For one part of the population stock market expectations are actually linked to investment in the stock index as predicted by economic theory. Higher expectations also lead to a larger investment in the stock index. This does not seem to be true for a second group in the population where the marginal effect centers around zero. This part of the population may follow different, e.g. heuristic decision rules in their portfolio choice and their stated beliefs do not appear to be a relevant constituent of the investment decision. 

This appearance of types is in line with other results from the portfolio choice literature such as \cite{GaudeckerME}. They establish a link between the precision of subjective beliefs and the predictive power of economic models. Whenever beliefs are rather crude and imprecise they are likely not determinants of a rational portfolio choice. For individuals with such beliefs, economic theory has a rather low power in predicting their stock market participation. 
This is in line with the finding here that for some part of the population their stated, subjective beliefs do not seem to influence their investment decision. \\

So far the finding indicates the existence of two, equally-large groups in the population. One group for which beliefs seem to have an impact on investment choice and one where it does not. 
Next I study heterogeneity of the random slope densities, i.e. consider estimates of $f_{B_1|X=x}$ evaluated at different testpoints $x$. This is interesting because the type distribution may vary across subpopulations with different observable characteristics $X$. 

To get an idea of which variables may drive the heterogeneity I report a variable importance measure for the density's shape and center, as outlined in Section \ref{Sec::CVetc}. 
The largest variable importance scores among the 75 control variables are reported in Table \ref{VI_measures}.
    \begin{table}[!htbp]
	\centering
	\begin{tabular}{l r r r r r r }
		& "age" & "daxnetto1" & "daxnetto2" & "prisk"  &   "isb011"    \\
		\hline\\[-1.0em]
		$VI_{shape}$  &  0.052 &  0.043 &  0.046 &  0.035 &  0.031     \\
		$VI_{mean}$ & 0.032 &  0.050 &   0.049 &  0.025 &   0.025  \\
	\end{tabular}
	\caption{Variable importance measures for those variables in $X$ with largest scores.}\label{VI_measures}
\end{table}
  
 There does not appear to be much variation in densities across different controls. The most importance is given to the age variable followed by "daxnetto1" and "daxnetto2" which are randomly selected information on past annual DAX returns that were presented to the survey respondents before the investment game. The other two are a measure of risk-aversion and a measure of self-assessed patience.  

Taking these results allows to investigate heterogeneity with respect to age and the historic information. 

Figure \ref{Plot::AgeHet} shows the heterogeneity in random slope densities for different age groups. Therefore the conditional density estimate is evaluated at three different testpoints. The variable age is varied but all other points are set to the respective sample median value of the variables. So $x$ is as in Figure \ref{Plot::MargDensity} except that age is varied.   
    \begin{figure}[!htbp]
	\centering
	\includegraphics[height=0.33\textheight, width=0.75\textwidth]{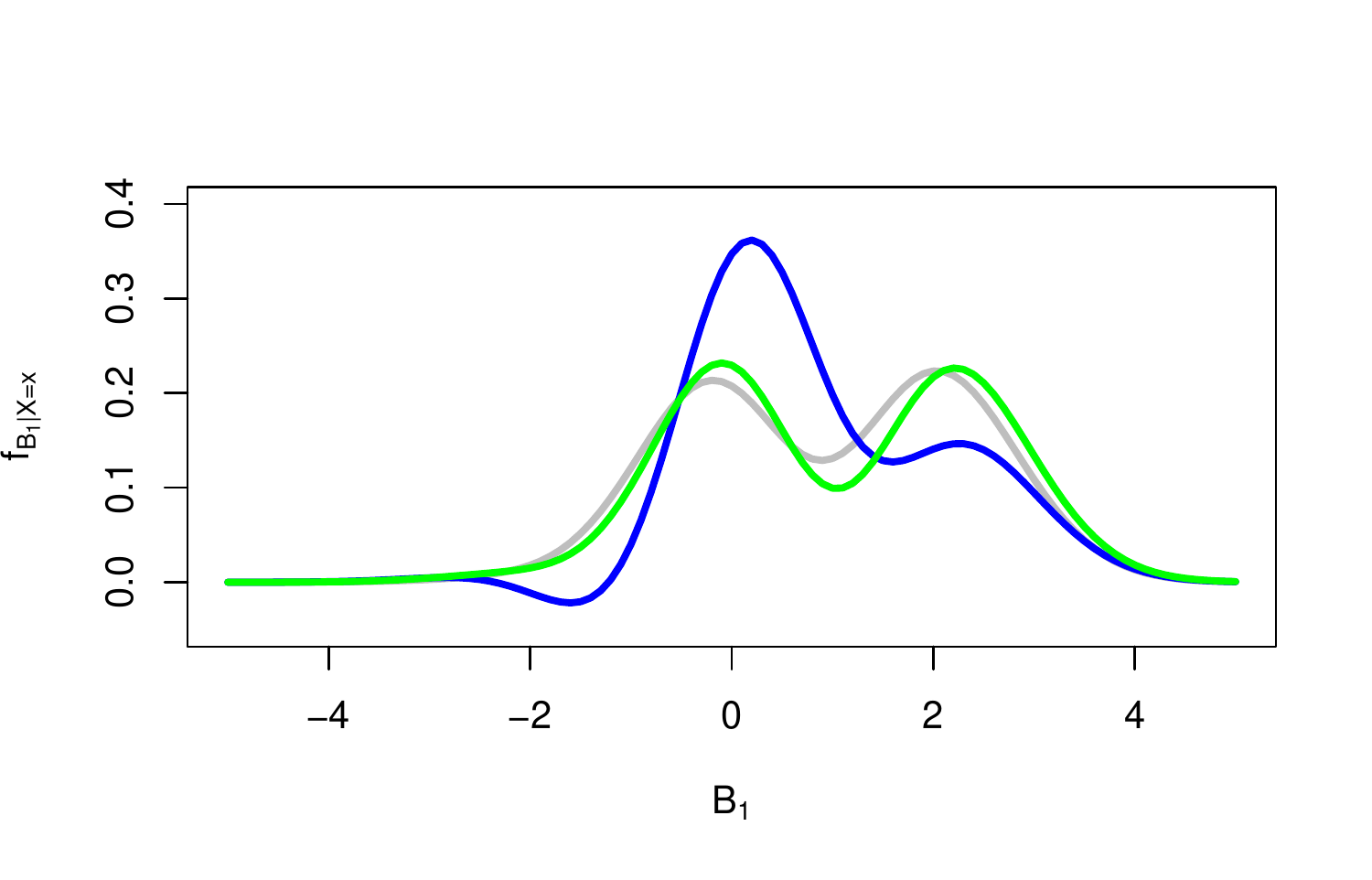}
	
	\caption{Estimate of the conditional density of $B_1|X=x$ for three different testpoints for $x$. The green line denotes the density for $age=30$, the grey line for $age=49$ and the blue line for $age=70$. The tuning parameters have been chosen as in Figure \ref{Plot::MargDensity}. } 
	\label{Plot::AgeHet}	 
\end{figure}
The most prominent descriptive fact here is that the type composure in the population seems to vary with age. For the young and medium aged subpopulations both types are equal in size. For the elder subpopulation fewer people behave according to economic theory. 

Next I also consider heterogeneity with respect to the historic information that has been displayed to the respondents. Again there are three testpoints. There is one testpoint where both historic informations have been very positive (return of 35\%), one where both informations are negative (return of -5\%) and one mixed with a positive first information and a negative second. The results are displayed in Figure \ref{Plot::AgeHet2}.
     \begin{figure}[!htbp]
 	\centering
 	\includegraphics[height=0.33\textheight, width=0.75\textwidth]{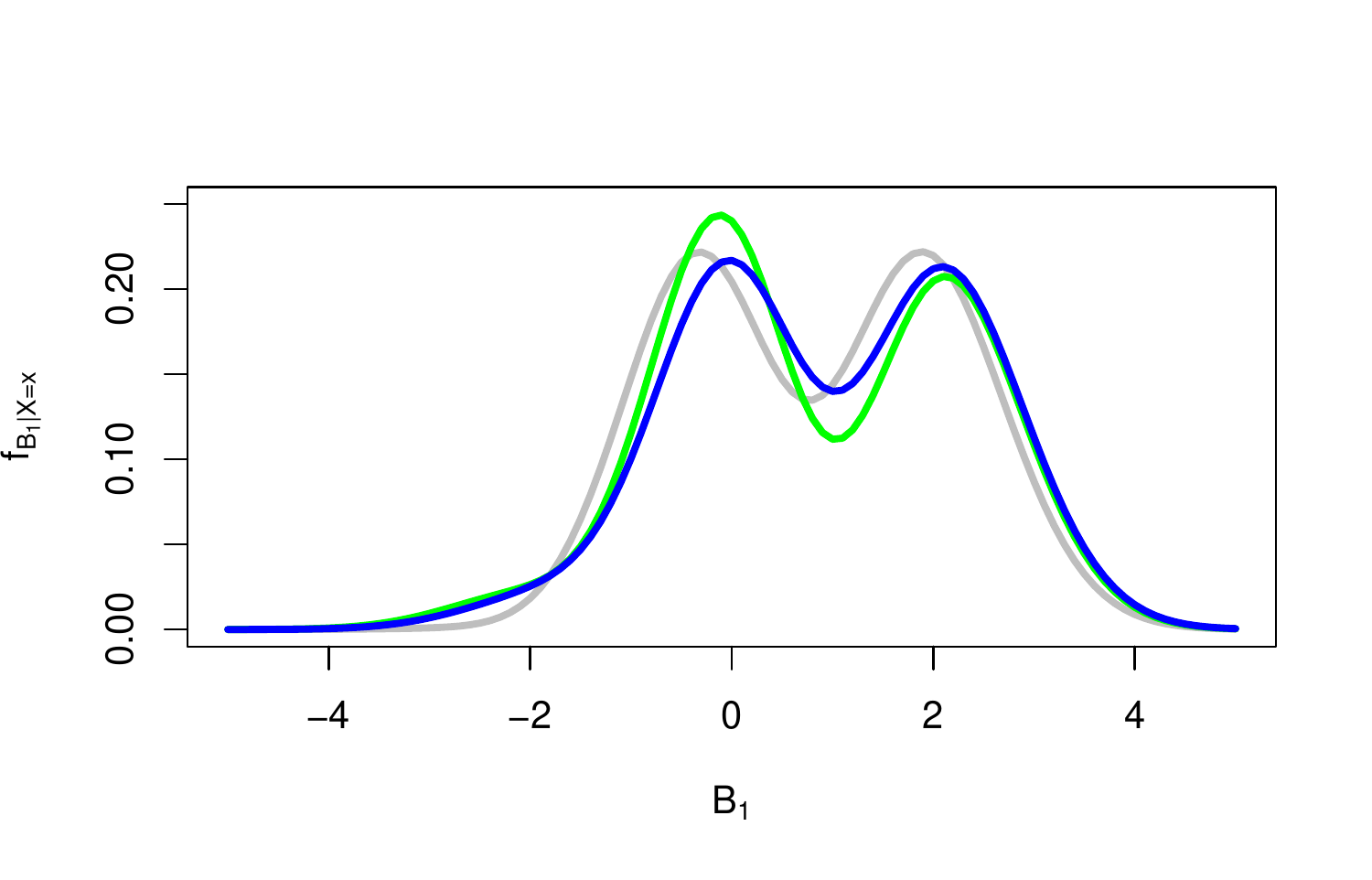}
 	
 	\caption{Estimate of the conditional density of $B_1|X=x$ for three different testpoints for $x$. The blue line denotes the density for $x$ with $daxnetto1=daxnetto2 = 35$, the grey line for $daxnetto1=daxnetto2 = -5$ and the green line for $daxnetto1=35$ and $daxnetto2 = -5$. The tuning parameters have been chosen as in Figure \ref{Plot::MargDensity}.} 
 	\label{Plot::AgeHet2}	 
 \end{figure}
Here there is no apparent or robust heterogeneity with respect to the historic information. 
We therefore stop the analysis and do not vary according to variables with a lower variable importance score than that of $daxnetto2$. 

The random coefficient analysis suggests the presence of two, roughly equally-sized types in the population. One group of individuals complies with economic theory in that their stock market expectation also explains their investment in a risky asset. A second group seems to follow different decision rules, their beliefs have no impact on their investment decision. 
Due to a possible correlation of beliefs and random coefficients this finding cannot be inferred from estimating standard, marginal random coefficient models.   

Regarding heterogeneity I find that the mixture of types in the populuation may depend on age but fail to uncover more interesting heterogeneity with the given data. 

It appears that the main determinants of type membership are unobservables that are not captured in the given data set.  
    \begin{figure}[!htbp]
	\centering
	\includegraphics[height=0.3\textheight, width=\textwidth]{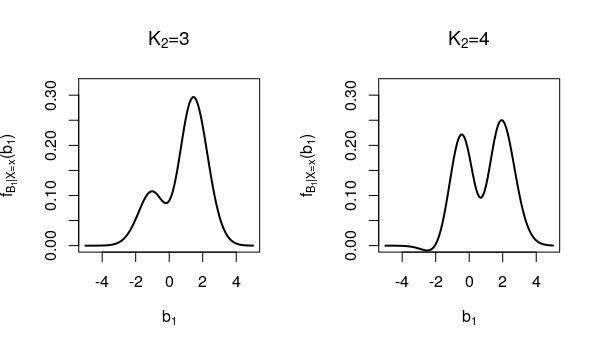} \\
	\includegraphics[height=0.3\textheight, width=\textwidth]{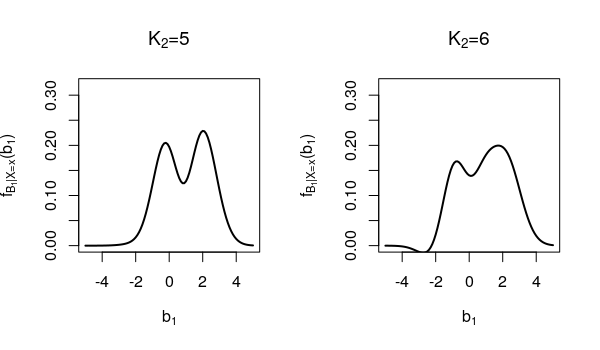}
	\includegraphics[height=0.3\textheight, width=\textwidth]{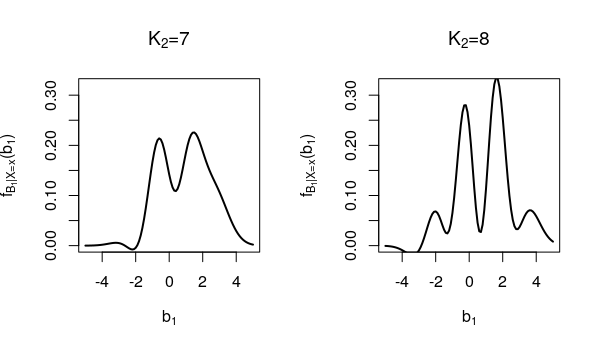}
	\caption{Estimates of the RC-density for various choices of $K_2$.} \label{Plot::Eyeballing}	 
\end{figure}

\section{Conclusion}\label{Sec::Conclusion}
The present paper discusses estimation of conditional random coefficient densities when the set of conditioning variables is large. The very general conditional RC model has rarely been studied in both theory and application. This paper provides a general sieve estimation strategy for estimating conditional RC densities. The approach enables the use of generic machine learning methods to estimate sieve coefficients in the presence of a large dimensional set of control variables. Therefore the estimator is applicable in many economic settings in which a continuous treatment variable is available.  
Theoretical results of the paper include convergence rate and inference results for the conditional sieve RC density estimator which combine asymptotic theories of sieve estimators and machine learning methods, in particular applying results on (honest) random forests. 

The finite sample properties of the estimator are illustrated in a Monte Carlo simulation study and an empirical application. The application reveals behavioral heterogeneity in an experimental portfolio choice task which is in line with recent empirical findings in the literature.

\appendix
\section{Proof of Theorems}
\begin{proof}[Proof of Lemma \ref{Lem::Ident}]
	Let $\phi_{Y|X}(t|x)=\Ex[\exp(itY)\;|\; X=x]$ denote the conditional characteristic function of $Y$ given $X=x$. The following holds
	\begin{align*}
	\phi_{Y|X, W}(t|x,w) &=  \Ex[\exp(itY)\;|\; X=x, W=w] \\
	&=  \Ex[\exp(it(B_0+B_1 W))\;|\; X=x, W=w]\\
	&= \Ex[\exp(i(t, tw)'(B_0, B_1))\;|\; X=x, W=w]\\
	&= \Ex[\exp(i(t, tw)'(B_0, B_1))\;|\; X=x]\\
	&=  \phi_{B_0, B_1 |X}(t, tw|x)
	\end{align*}
	which is in fact already enough to point identify the probability distribution of $B\;|\;X=x$. By varying both $t$ and $w$ it is possible to evaluate the characteristic function of $B\;|\;X=x$ at any point in $\mathbb{R}^2$. See the proof of Lemma 1 in \cite{masten17_RC} and the references therein for details. Here is where Assumption \ref{A::Ident1} is required in that the support of $W\;|\;X=x$ is the entire real line $\mathbb{R}$. 
	
	The main interest in practical applications is in identifying the density function $f_{B|X=x}$ which follows from applying the inverse Fourier transform to $\phi_{B_0, B_1 |X=x}(t, tw)$. 
	The Fourier transformation $\mathcal{F}$ and the inverse Fourier transformation $\mathcal{F}^{-1}$ are defined as
	\begin{align*}
	(\mathcal{F}f)(t) & =\int_{\mathbb{R}^d} \exp(it'a)f(a)da \\
	(\mathcal{F}^{-1}g)(a) &= \frac{1}{(2\pi)^d}\int_{\mathbb{R}^d} \exp(-ia't)g(t)dt 
	\end{align*}
	for some functions $f,g: \mathbb{R}^d \to \mathbb{R}$ and $\mathcal{F}: \mathbb{R}^d \to \mathbb{C}^d$ and $\mathcal{F}^{-1}: \mathbb{C}^d \to \mathbb{R}^d$. The Fourier transform generally links the characteristic function of a random variable to its density function, in particular here $\phi_{B_0, B_1 |X}(t, tw|x) = (\mathcal{F}f_{B_0, B_1 |X=x})(t, tw)$.
	
	From this we can infer the following
	\begin{align*}
	f_{B|X}(b|x) & = \frac{1}{(2\pi)^2}\int_{\mathbb{R}^2} \exp(-ib's)(\mathcal{F} f_{B|X=x} )(s)ds \\
	&=  \frac{1}{(2\pi)^2}\int_{\mathbb{R}^2} |t| \exp(-ib'(t,tw))(\mathcal{F} f_{B|X=x} )(t, tw)dt dw \\
	&=  \frac{1}{(2\pi)^2}\int_{\mathbb{R}^2} |t|\exp(-ib'(t,tw))(\mathcal{F} f_{B|X=x} )(t, tw)dt dw \\
	&=  \frac{1}{(2\pi)^2}\int_{\mathbb{R}^2} |t| \exp(-ib'(t,tw))\phi_{B_0, B_1 |X}(t, tw|x)dt dw \\
	&=  \frac{1}{(2\pi)^2}\int_{\mathbb{R}^2} |t| \exp(-ib'(t,tw))\phi_{Y|X, W}(t|x,w)dt dw 
	\end{align*}
	
	which establishes identification of the density function $f_{B|X=x}$.
\end{proof}

\begin{proof}[Proof of Theorem \ref{Thm::Rate}]
	Let $\mynorm{\cdot}_E$ denote the euclidean norm of a (complex-) vector and define $\mynorm{f-g}=\int_{\mathbb{R}^2} |f(a)-g(a)|^2da$ and $\mynorm{f-g}_{\nu,\mu}=\int_{\mathbb{R}^2} |\mathcal{F}f(t,w)-\mathcal{F}g(t,w)|^2d\nu(t)d\mu(w)$ for arbitrary functions $f,g:\mathbb{C}^2 \to \mathbb{R}$.
	 Consider the following decomposition

	\begin{align*}
	\mynorm{\widehat{f}_{B|X=x} - {f}_{B|X=x}}^2 & \leq \mynorm{ \widehat{f}_{A|X=x} - f_{A|X=x} }^2 + \mynorm{ f_{A|X=x}(\cdot - \widehat \beta(x)) - f_{A|X=x}(\cdot -\beta(x))}^2 \\
	& \leq  \mynorm{\widehat{f}_{A|X=x} - \widetilde{f}_{A|X=x}}^2 + \mynorm{\widetilde{f}_{A|X=x} - {f}_{A|X=x}}^2  \\
	& \quad + \mynorm{ f_{A|X=x}(\cdot - \widehat \beta(x)) - f_{A|X=x}(\cdot -\beta(x))}^2 \\
	& = A+ B+ C
	\end{align*}
	The proof begins by examining summand $B$. To this end recall that
	\begin{align*}
	P_K{f}_{A|X=x} = \arg \min_{\phi \in \mathcal{B}_K} \mynorm{\phi - {f}_{A|X=x} }
	\end{align*}
	which is the $L_2$-projection of ${f}_{A|X=x}$ on the sieve space $\mathcal{B}_K$.  It further holds for every $x\in \mathcal{X}$,
	
	\begin{align*}
	\mynorm{\widetilde{f}_{A|X=x} - {f}_{A|X=x}}^2 &\leq\mynorm{\widetilde{f}_{A|X=x} - P_K {f}_{A|X=x}}^2  + \mynorm{P_K{f}_{A|X=x} - {f}_{A|X=x}}^2 \\
	& \leq {\tau_K}^{-1} \mynorm{\mathcal{F}\widetilde{f}_{A|X=x} - \mathcal{F}P_K{f}_{A|X=x} }^2_{v,\mu} + O(K^{-\alpha})  \\
	& \leq  {\tau_K}^{-1}\Big[ \mynorm{\mathcal{F}\widetilde{f}_{A|X=x} - \mathcal{F}{f}_{A|X=x} }^2_{v,\mu} +  \mynorm{\mathcal{F}{f}_{A|X=x} - \mathcal{F}P_K{f}_{A|X=x} }^2_{v,\mu} \Big] + O(K^{-\alpha}) \\
	& \leq {\tau_K}^{-1} \mynorm{\mathcal{F}{f}_{A|X=x} - \mathcal{F}P_K{f}_{A|X=x} }^2_{v,\mu}  + O(K^{-\alpha}) \\
	& = O(K^{-\alpha})
	\end{align*}
	where we have used the link condition $\mynorm{\mathcal{F}{f}_{A|X=x} - \mathcal{F}\Pi_K{f}_{A|X=x} }^2_{v,\mu}=O(\tau_K \mynorm{\Pi_K{f}_{A|X=x} - {f}_{A|X=x}}^2)$ and the fact that
	\begin{align*}
	\widetilde{f}_{A|X=x}=\arg \min_{\phi \in \mathcal{B}_K} \mynorm{\mathcal{F}\phi - \mathcal{F}{f}_{A|X=x} }^2_{v,\mu}.
	\end{align*} 
	
	\noindent Next consider the first summand $A$. It holds that  
	\begin{align}
	\mynorm{\widehat{f}_{A|X=x} - \widetilde{f}_{A|X=x}}^2 = & \mynorm{q^K(\cdot)'\widehat Q^{-1} \widehat \Pi_{dm}(x) - q^K(\cdot)'Q^{-1} \Pi(x) }^2 \nonumber \\
	\leq &  \mynorm{q^K(\cdot)'(\widehat Q^{-1} - Q^{-1})\Pi(x)}^2 + \mynorm{q^K(\cdot)'(\widehat Q^{-1} - Q^{-1}) (\widehat \Pi_{dm}(x) - \Pi(x))}^2  \nonumber \\
	 & + \mynorm{q^K(\cdot)'\widehat Q^{-1}(\widehat \Pi_{dm}(x) - \Pi(x))}^2 \nonumber \\
	=:& I + II + III \nonumber
	\end{align}
	
	where in the following we consider each term separately. We begin with term III where we have 
	\begin{align*}
	\mynorm{q^K(\cdot)'Q^{-1}(\widehat \Pi_{dm}(x) -  \Pi(x))}^2 & = [\widehat \Pi_{dm}(x) -  \Pi(x)]'Q^{-1}\left(\int_{\mathbb{R}^2} q^K(b)q^K(b)'db \right)  Q^{-1}[\widehat \Pi_{dm}(x) -  \Pi(x)]\\
	& \lesssim \mynorm{\widehat \Pi_{dm}(x) -  \Pi(x)}_E^2 \mynorm{Q^{-1}}^2 \\
	& \lesssim \tau^{-2}_K \Big[\mynorm{\widehat \Pi_{dm}(x) -  \Pi_{dm}(x) }^2_E  + \mynorm{ \Pi_{dm}(x) -  \Pi(x)  }^2_E    \Big] \\
	& \lesssim \tau^{-2}_K \Big[ K \cdot O_p( n^{-2\varphi}) +  O_p(K \cdot n^{-2\varphi})   \Big] \\
	& = O_p(\tau^{-2}_K \frac{K}{n^{2\varphi}} ).
	\end{align*}
	and made use of the sample splitting rule and the convergence rate of ML-estimates in Assumption \ref{A::MLRates} (i). Without sample splitting the behavior of $\mynorm{\widehat \Pi_{dm}(x) -  \Pi_{dm}(x)}_{E}$ cannot be established. The behavior of $\mynorm{ \Pi_{dm}(x) -  \Pi(x)  }^2_E$ follows from Lemma \ref{Lem1::Appendix} (ii).
	

Next, consider the term $I$. We have
	\begin{align*}
  \mynorm{q^K(\cdot)'(\widehat Q^{-1} - Q^{-1})\Pi(x)}^2  & = \Pi(X)'(\widehat Q^{-1} - Q^{-1})\left(\int_{\mathbb{R}^2} q^K(b)q^K(b)'db \right)(\widehat Q^{-1} - Q^{-1})\Pi(x)\\
	& \lesssim \mynorm{\Pi(x)}_E^2 \cdot  \mynorm{\widehat Q^{-1} - Q^{-1} }^2 \\
	& \lesssim K \cdot O_p\left(\frac{K\tau^{-1}_K\log(K)}{n}\right) \\
	& \lesssim o_p(\tau^{-2}_K \frac{K}{n^{2\varphi}} )
	\end{align*}
	which holds by Assumption \ref{A::Rate} (iv) and applying Rudelsons LLN, as in the second part of Lemma 6.2. in \cite{BellChern_Series} to $\mynorm{\widehat Q^{-1} - Q^{-1}}$. The last inequality holds by the rate restriction in Assumption \ref{A::MLRates} (iii).  \\
	
	\noindent It remains to analyze II. Under the same reasoning as above we obtain
	\begin{align*}
	& \quad  \mynorm{q^K(\cdot)'(\widehat Q^{-1} - Q^{-1}) (\widehat \Pi_{dm}(X) -  \Pi_{dm}(X))} \\
	& \lesssim  \mynorm{\widehat \Pi_{dm}(X) -  \Pi(X)}_E^2 \mynorm{\widehat Q^{-1} - Q^{-1} }^2 \\
	&\lesssim O_p\left(\frac{K}{n^{2\varphi}}\right)  \cdot O_p\left(\frac{K\tau^{-1}_K\log(K)}{n}\right) \\
	& \lesssim o_p\left(\tau^{-2}_K \frac{K}{n^{2\varphi}}\right)
	\end{align*}
	and collecting terms we can conclude that $A=  O_p(\tau^{-2}_K K/n^{2\varphi} )$.
    
    Finally it remains to consider $C$.  
	There exists some $\tau \in (0,1)$ such that the following holds
	\begin{align*}
	\mynorm{ f_{A|X}(\cdot - \widehat \beta(X)) - f_{A|X}(\cdot -\beta(X))} & \leq    \int \mynorm{ \nabla f_{A|X}(b-\xi)}db \cdot \mynorm{\widehat \beta(X)- \beta(X) }_E^2 \\
	& = O_p(1) \cdot O_p(n^{-2\varphi}) =o_p(\tau^{-2}_K \frac{K}{n^{2\varphi}})
	\end{align*}
	where $\xi=\beta(X)(1-\tau) + \tau \widehat{\beta}(X)$ and the last bound following
   from Assumption \ref{A::Rate} (v) and \ref{A::MLRates}. 
	This establishes the final result of Theorem \ref{Thm::Rate}. 
\end{proof}

\begin{proof}
	In the case outlined in the corollary, we have 
	\begin{align*}
	Q & = \Ex\left[\int_{\mathbb{R}} \mathcal{F}q^K(-t, -t(W-g(X)))\mathcal{F}q^K(-t, -t(W-g(X)))'d\nu(t) \right] \\
	\widehat{Q} & = \sum_{i=1}^{\mathcal{R}}  \int_{\mathbb{R}}  \mathcal{F}q^K(-t, -t(W_i-\widehat{g}(X_i))\mathcal{F}q^K(-t, -t(W_i-\widehat{g}(X_i)))'d\nu(t) \\
	\widetilde{Q} & = \Ex\left[\int_{\mathbb{R}} \mathcal{F}q^K(-t, -t (W-\widehat{g}(X)))\mathcal{F}q^K(-t, -t(W-\widehat{g}(X)))'d\nu(t)\right] .
	\end{align*}
	The main part is to consider the behavior of $\mynorm{\widehat{Q} -Q}$. It holds that
	\begin{align*}
\mynorm{\widehat{Q} -Q} &= \mynorm{\widehat{Q} - \widetilde{Q}} + \mynorm{\widetilde{Q} - Q} \\
& = O_p\left(\sqrt{\frac{K\tau_K^{-1}\log(K)}{n}}\right) + \mynorm{\widetilde{Q} - Q}  
	\end{align*}
again by Rudelson's LLN. For the second part $\mynorm{\widetilde{Q} - Q}$ it suffices to check the quantity
\begin{align}\label{Q_ineq}
&|| \mathcal{F}q^K(-t, -t (W-\widehat{g}(X)))\mathcal{F}q^K(-t, -t(W-\widehat{g}(X))) \\
- & \mathcal{F}q^K(-t, -t (W-g(X)))\mathcal{F}q^K(-t, -t(W-g(X))) ||.  \nonumber
\end{align}	
For arbitrary complex vectors $a,b$ it holds that
\begin{align*}
\mynorm{aa'-bb'}& = \mynorm{(a-b)(a-b)'+ (a-b)b'+b(a-b)'} \\
  & \leq 2\cdot \mynorm{a-b}+ 2\cdot \mynorm{b}\cdot \mynorm{a-b}
\end{align*}
and applying this to (\ref{Q_ineq}) leads to 
\begin{align*}
(\ref{Q_ineq}) & \leq 2\cdot (1 + \mynorm{\mathcal{F}q^K(-t, -t (W-g(X)))}) \cdot  \mynorm{\mathcal{F}q^K(-t, -t (W-\widehat{g}(X)) - \mathcal{F}q^K(-t, -t (W-g(X))}  \\
& \leq 2\cdot (1 + \mynorm{\mathcal{F}q^K(-t, -t (W-g(X)))}) \cdot \mynorm{ D\mathcal{F}q^K(\xi)}\cdot \mynorm{\widehat g (X)-g(X)} 
\end{align*}
which implies that 
\begin{align*}
\mynorm{\widetilde{Q} - Q} \lesssim \sqrt{K} \cdot K \cdot O_p(n^{-2\varphi}) 
\end{align*}
and thus 
\begin{align*}
\mynorm{\widehat{Q} -Q} =  O_p\left(\sqrt{\frac{K\tau_K^{-1}\log(K)}{n}}\right) + O_p \left( \frac{K^{3/2}}{ n^{2\varphi}}  \right)
\end{align*}
The difference to the proof of Theorem \ref{Thm::Rate} is only in checking terms I, II and III. 
By applying Lemma \ref{Lem1::Appendix} (i) it holds that 
\begin{equation*}
III = O_p(\tau_K^{-2} \frac{K^2}{n^{2\varphi}})
\end{equation*}
and further from the rate of $\mynorm{\widehat{Q} -Q}$ above and the rate restriction stated in the Corollary that
\begin{equation*}
I = o_p(\tau_K^{-2} \frac{K^2}{n^{2\varphi}}).
\end{equation*}
From III and I it is apparent that II is asymptotically negligible, which leads to the stated result.
\end{proof}

\begin{proof}[Proof of Lemma \ref{Lem::ForestAN}]
For any single (honest) random forest predictor $\widehat{\Pi}_{dm,k}(x)$ Theorem 1 of \cite{AtheyWager18} establishes its asymptotic normality under the assumptions stated in the Theorem itself. 
For the proof of Lemma \ref{Lem::ForestAN} it suffices to adapt their steps to the case $q^K(b)'Q^{-1}(\widehat{\Pi}_{dm,1}(x), \dots, \widehat{\Pi}_{dm,K}(x)$. To simplify notation for the remainder of the proof $q^K=q^K(b)$ and $\widehat{\Pi}_{dm,k}=\widehat{\Pi}_{dm,k}(x)$.
The proof proceeds treating $\widehat \Pi_{dm, k}$ as a pure forest estimator. Applying the integration to obtain the true $\widehat \Pi_{dm, k}$ of (\ref{Final Pi_dm}), does not change the derivations, as integration is a linear, monotonic operator and the theory in \cite{AtheyWager18} goes through. \\
 
Let $\overset{\circ}{\widehat{\Pi}}_{dm,k}$ denote the Hajek projection of the forest predictor  	
and under a slight abuse of notation let $\overset{\circ}{\widehat{\Pi}}_{dm} = (\overset{\circ}{\widehat{\Pi}}_{dm,1}, \dots, \overset{\circ}{\widehat{\Pi}}_{dm,K}  )'$.  \\

In broad steps the proof of \cite{AtheyWager18} proceeds by checking that for a given forest predictor 
$\widehat{\Pi}_{dm,k}$ it holds that
\begin{align}
  & \frac{\overset{\circ}{\widehat{\Pi}}_{dm,k} - \Ex[\overset{\circ}{\widehat{\Pi}}_{dm,k}]}{\sigma_{dm,k}}  \overset{d}{\to} N(0,1) \\
 & \Ex\left[\left( \widehat{\Pi}_{dm,k} - \overset{\circ}{\widehat{\Pi}}_{dm,k} \right)^2\right]  /\sigma^2_{dm,k} \to 0 \\
 & \frac{\Ex[\widehat{\Pi}_{dm,k}] -\Pi_{dm,k} }{\sigma_{dm,K}}  \to 0
\end{align}
where (A.1) and (A.2) are shown in the proof of Theorem 8 and (A.3) in the proof of Theorem 1 of \cite{AtheyWager18}. The quantity $\sigma_{dm,k}$ is in fact the standard deviation of the Hajek projection.  

I follow along their steps and show that the following holds under the Assumptions stated in Lemma \ref{Lem::ForestAN}: 
\begin{align*}
I := & \frac{q^KQ^{-1}(\overset{\circ}{\widehat{\Pi}}_{dm} - \Ex[\overset{\circ}{\widehat{\Pi}}_{dm}])}{\mynorm{q^KQ^{-1}\Sigma_n}} \overset{d}{\to} N(0,1) \\
II := &\Ex\left[\left(q^KQ^{-1}\left(\widehat{\Pi}_{dm} - \overset{\circ}{\widehat{\Pi}}_{dm} \right)\right)^2\right]  /\mynorm{q^KQ^{-1}\Sigma_n}^2 \to 0 \\
III := & \frac{q^KQ^{-1}(\Ex[\widehat{\Pi}_{dm}] -\Pi_{dm}) }{\mynorm{q^KQ^{-1}\Sigma_n}} \to 0
\end{align*}
which taken together implies that $q^KQ^{-1}(\widehat{\Pi}_{dm} - \Pi_{dm})/\mynorm{q^KQ^{-1}\Sigma_n} $ is asymptotically normal. \\

Before we need to introduce and adapt some of the notation from the proofs of \cite{AtheyWager18}. Let $s$ denote the subsample size used to construct the random forest from single tree predictors $\widehat{T}=(\widehat{T}_1 ,\dots, \widehat{T}_K)$ where $\widehat T_k=\widehat T_k(x; \mathcal{R}_k)$
which is a single tree predictor for the conditional expectation $\Ex[T_k(W-\widehat g(X), Y-\widehat m(X,W))|X=x]$ making use of the data points in the respective sample $\mathcal{R}_k$.

We begin with part I. 
Plugging in the expression for the Hajek projection of the random forest on page 53 of the supplemental material of \cite{AtheyWager18} we obtain the identity
\begin{align*}
q^KQ^{-1}(\overset{\circ}{\widehat{\Pi}}_{dm} - \Ex[\overset{\circ}{\widehat{\Pi}}_{dm}]) = \frac{s\cdot K}{n}\sum_{i=1}^{n/K} q^KQ^{-1}(\Ex[\widehat T |\mathcal{R}_i] - \Ex[\widehat T])
\end{align*}
where $\Ex[\widehat T |\mathcal{R}_i] = (\Ex[\widehat T_1 |\mathcal{R}_{1,i}], \dots, \Ex[\widehat T_K |\mathcal{R}_{K,i}])'$ and $\mathcal{R}_{k,i}$ is the $i$-th observation in sample $\mathcal{R}_{k}$.  Further,
\begin{align*}
\mynorm{q^KQ^{-1}\Sigma_n}= \frac{s\cdot K }{n} \sqrt{ \sum_{i=1}^{n/K} q^KQ^{-1}Var(\widehat T)  Q^{-1}q^K}
\end{align*} 
where $Var(\widehat T)=diag(Var(\widehat{T}_1), \dots, Var(\widehat{T}_K))$ and which
holds from applying the identity on the last line of page 52 and thus we can write
\begin{align*}
I = \frac{\sum_{i=1}^{n/K} q^KQ^{-1}(\Ex[\widehat T |\mathcal{R}_i] - \Ex[\widehat T]) }{\sqrt{ \sum_{i=1}^{n/K} q^KQ^{-1}Var(\widehat{T})  Q^{-1}q^K} }
\end{align*}
and establish the asymptotic normality of I by checking Lyapunov's condition
\begin{align}\label{Lyapunov}
 \frac{\sum_{i=1}^{n/K} \Ex[| q^KQ^{-1}(\Ex[\widehat T |\mathcal{R}_i] - \Ex[\widehat T])|^{2+\delta}] }{\left(\sum_{i=1}^{n/K} q^KQ^{-1}Var(\widehat{T})  Q^{-1}q^K \right)^{1+\delta/2}} \to 0
\end{align}
For the numerator we have by Cauchy-Schwarz
\begin{align*}
\sum_{i=1}^{n/K} \Ex[| q^KQ^{-1}(\Ex[\widehat T |\mathcal{R}_i] - \Ex[\widehat T])|^{2+\delta}] & \leq \mynorm{q^KQ^{-1}}^{2+\delta} \cdot \sum_{k=1}^{K} \sum_{i=1}^{n/K} \Ex[|\Ex[\widehat T_k |\mathcal{R}_{k,i}] - \Ex[\widehat T_k]|^{2+\delta}  ]  \\
& \leq \mynorm{q^KQ^{-1}}^{2+\delta} \cdot K \cdot \sum_{i=1}^{n/K} \max_{k}  \Ex[|\Ex[\widehat T_k |\mathcal{R}_{k,i}] - \Ex[\widehat T_k]|^{2+\delta}] 
\end{align*}
where the last inequality is due to the last display in the proof of Theorem 8.
The denominator satisfies
\begin{align*}
 \left(\sum_{i=1}^{n/K} q^KQ^{-1}Var(\widehat{T})Q^{-1}q^K \right)^{1+\delta/2} & \geq 
 \mynorm{q^KQ^{-1}}^{2+\delta}  \cdot \left(\sum_{i=1}^{n/K} \min_{k} Var(\widehat T_k)\right)^{1+\delta/2} 
\end{align*}   
which follows from the last steps of the proof on page 54. 
Define
\begin{align*}
k^* &:= \arg\max_{k \in K}  \Ex[|\Ex[\widehat T_{k} |\mathcal{R}_{k,i}] - \Ex[\widehat T_{k}]|^{2+\delta}]   \\
\overline{k} &:= \arg\max_{k \in K}  Var(\widehat T_k) \\
\underline{k} &:= \arg\min_{k \in K}  Var(\widehat T_k)
\end{align*}
then following the proof of Theorem 8 one can conclude that 
\begin{align*}
(\ref{Lyapunov}) \leq \frac{K\cdot \sum_{i=1}^{n/K} \Ex[|\Ex[\widehat T_{k^*} |\mathcal{R}_{k^*,i}] - \Ex[\widehat T_{k^*}]|^{2+\delta}]  }{\left(\sum_{i=1}^{n/K} Var(\widehat T_{k^*})\right)^{1+\delta/2} } \cdot \frac{\sum_{i=1}^{n/K} Var\left(\widehat T_{\overline k}\right)^{1+\delta/2}}{\sum_{i=1}^{n/K} Var\left(\widehat T_{\underline k}\right)^{1+\delta/2}} \ \leq K \cdot r_p(n/K)^{-\delta/2}
\end{align*}
which holds by the last display in the proof of Theorem 8 in the supplemental material of \cite{AtheyWager18} and Assumption \ref{A::Norm_I} (ii) which implies $\sum_{i=1}^{n/K} Var\left(\widehat T_{\overline k}\right)^{1+\delta/2}/\sum_{i=1}^{n/K} Var\left(\widehat T_{\underline k}\right)^{1+\delta/2}=O(1)$. 
Then (\ref{Lyapunov}) tends to zero by the rate restriction stated in Lemma \ref{Lem::ForestAN} which establishes asymptotic normality of I. \\

Then consider II. By applying Cauchy Schwarz and the lower bound for sieve variance we obtain
\begin{align*}
II \leq & \frac{\mynorm{q^KQ^{-1}}^2}{\mynorm{q^KQ^{-1}\Sigma_n}^2} \Ex[\mynorm{\widehat{\Pi}_{dm} - \overset{\circ}{\widehat{\Pi}}_{dm}}^2 ] \\
\leq &  \frac{1}{\min_{k} \sigma_{dm,k}^2} \sum_{k=1}^{K} \Ex[(\widehat{\Pi}_{dm,k} - \overset{\circ}{\widehat{\Pi}}_{dm,k})^2 ]  \\
\leq & K \cdot r_p(n/K)^{-1} \to 0
\end{align*}
which holds by the same reasoning as in the beginning of the proof of Theorem 8 in the supplement of \cite{AtheyWager18} and by the rate restriction stated in Lemma \ref{Lem::ForestAN}.

It remains to consider III. By the same reasoning as before we obtain
\begin{align*}
III &\leq \frac{\mynorm{\Ex[\widehat{\Pi}_{dm}] -\Pi_{dm}}}{ \min_{k} \sigma_{dm,k}} \\
 & \leq \frac{\sqrt{K}\cdot \max_k  \Ex[\widehat{\Pi}_{dm,k}] -\Pi_{dm,k} }{\min_{k} \sigma_{dm,k} } \\
 & \lesssim \sqrt{K} \cdot  \left( \frac{n}{K}\right)^{\frac{1}{2}\beta^\ast} =O\left(\sqrt{\frac{K}{ r_p(n/K)^{-\beta^\ast/b}}} \right) 
\end{align*}
where $\beta^\ast := 1+\epsilon-b/\beta_{\min}$. 
The last inequality follows from the Proof of Theorem 1 on page 40 of the supplement of \cite{AtheyWager18} and under these conditions $\beta^\ast< 0$ and $b/\beta_{\min} > 0$.
Under the rate restrictions in the Lemma the right hand-side above converges to zero which concludes the proof. 

\end{proof}	

\begin{proof}[Proof of Theorem \ref{thm::Inference}]
The proof begins with the following decomposition
\begin{align*}
& \widehat{f}_{B|X}(b|x) - f_{B|X}(b|x) \\
=& \underbrace{\widehat{f}_{A|X}(b- \widehat\beta(x)|x) - \widehat{f}_{A|X}(b- \beta(x)|x)}_{I} +  \underbrace{\widehat{f}_{A|X}(b- \beta(x)|x) - q^K(b-\beta(x))'Q^{-1}\Pi_{dm}(x)}_{II} \\
 + & \underbrace{q^K(b-\beta(x))'Q^{-1}\Pi_{dm}(x) -  q^K(b-\beta(x))'Q^{-1}\Pi(x)}_{III} \\
 + & \underbrace{q^K(b-\beta(x))'Q^{-1}\Pi(x) - f_{A|X}(b-\beta(x)|x)}_{IV}
\end{align*}
and proceeds by checking the individual terms separately.

For I it holds that 
\begin{align*}
I \leq & |\widehat{f}_{A|X}(b- \widehat\beta(x)|x) - \widehat{f}_{A|X}(b- \beta(x)|x) | \\
   \leq & \mynorm{D\widehat{f}_{A|X}(b -\beta(x)-(1- \tau)(\widehat{\beta}(x) -\beta(x))}\cdot \mynorm{\widehat{\beta}(x) -\beta(x)} \\
   = & o_p(v_n(b,w))
\end{align*}
for some $\tau \in (0,1)$ by consistency of $\widehat{f}_{A|X}$ and Assumption \ref{A::Rate} (v). This is due to the fact that $\mynorm{\widehat{\beta}(x) -\beta(x)}=r_p(n)^{-1}$ as $\widehat{\beta}$ is calculated on a sample proportional to $n$ and thus by the rate of $v_n(b,w)$ it always holds that $I= o_p(v_n(b,w))$.
For II it holds by Assumption \ref{A::Norm_I} (i) that
\begin{align*}
II/ v_n(b, x)  \overset{d}{\rightarrow}  N(0,1)
\end{align*}
For III we have from Lemma \ref{Lem1::Appendix} (ii)
that 
\begin{align*}
III \lesssim_P \tau^{-1}_K\sqrt{K}\cdot \sqrt{K/n^{2\varphi}}
\end{align*}
and thus
\begin{align*}
III/v_n(b,w) \lesssim_P \frac{\tau^{-1}_K\sqrt{K} \sqrt{K/n^{2\varphi}} }{\tau^{-1}_K\sqrt{K}r_p(n/K)^{-1}} =\frac{K}{n}=o(1)
\end{align*}

and thus $III/v_n(b,x) =o_p(1)$. \\


Finally for IV 
\begin{align*}
\left( P_Kf_{A|X}(b-\beta(x)|x) - f_{A|X}(b-\beta(x)|x)\right) = o(v_n(b,w))
\end{align*}
by Assumption \ref{A::Norm_II} (i) the approximation error is negligible compared to $v_n$. \\

For the final part of the statement it remains to show
\begin{align*}
\left| \frac{\widehat{v}_n(b,w) }{v_n(b,w)} -1 \right|=o_p(1) 
\end{align*}
Let $s(b,x)' = q^K(b-\beta(x))'Q^{-1}$ and $\widehat{s}(b,x)'=q^K(b - \widehat{\beta}(x))'\widehat{Q}^{-1}$ it holds that 
\begin{align*}
\left|\widehat v_n(b,w)- v_n(b,w)\right| & \leq \left|\mynorm{\widehat{s}(b,x)'\widehat{\Sigma}_n(x)}   - \mynorm{s(b,x)'\Sigma_n(x)}\right|\\
& \leq \mynorm{ \widehat{s}(b,x)'\widehat{\Sigma}_n(x) - s(b,x)'\Sigma_n(x)   } \\
& \leq \mynorm{ [\widehat{s}(b,x) - s(b,x)]'\widehat{\Sigma}_n(x) } +\mynorm{s(b,x)'(\widehat{\Sigma}_n(x) - \Sigma_n(x)) } \\
& \leq \max_{k} \widehat{\sigma}_{dm,k} \cdot \mynorm{\widehat{s}(b,x) - s(b,x) } + \max_{k} |\widehat{\sigma}_{dm,k} - \sigma_{dm,k} |\cdot \mynorm{s(b,x)}
\end{align*}
by the triangle inequality and the fact that $\widehat{\Sigma}_n(x),\Sigma_n(x)$ is a diagonal matrix. 

Further
\begin{align*}
\mynorm{\widehat{s}(b,x)- s(b,x)} & = \mynorm{ q^K(b - \widehat{\beta}(x))'\widehat{Q}^{-1} -  q^K(b-\beta(x))'Q^{-1}} \\
&\leq \mynorm{ [q^K(b - \widehat{\beta}(x)) -q^K(b-\beta(x))  ]'\widehat{Q}^{-1} + q^K(b-\beta(x))'(\widehat{Q}^{-1}-Q^{-1})  } \\
&\leq \mynorm{Dq^K(b-\beta(x)-(1- \tau)(\widehat{\beta}(x) -\beta(x)))}\cdot \mynorm{\widehat{Q}^{-1}}\cdot \mynorm{\widehat{\beta}(x) -\beta(x)}  \\
&+ \mynorm{q^K(\cdot-\beta(x))}\cdot\mynorm{\widehat{Q}^{-1}-Q^{-1}} \\
&\lesssim_P \sqrt{K}\tau_K^{-1}\cdot r_p(n)^{-1} + \sqrt{K}\cdot \sqrt{K\tau_K^{-1}\log(K)/n }
\end{align*}
Summarizing, by the properties of $v_n(b,x)$ it holds that 
\begin{align*}
\left|\frac{\widehat{v}_n(b,w)}{v_n(b,w)} -1 \right| & \lesssim \frac{\max_{k} \widehat{\sigma}_{dm,k}}{\max_{k} \sigma_{dm,k}} \cdot \frac{\max_{k}\sigma_{dm,k}}{\min_{k} \sigma_{dm,k}}\cdot \frac{\mynorm{\widehat{s}(b,x)-s(b,x)}}{ \sqrt{K}\tau_K^{-1}}  \\
& +  \frac{ \max_k | \widehat{\sigma}_{dm,k}-\sigma_{dm,k}|  }{\sigma_{dm,k^*} } \cdot \frac{\sigma_{dm,k^*} }{\min_k \sigma_{dm,k} } \cdot\frac{\mynorm{s(b,x)}}{\sqrt{K} \tau_K^{-1}} \\
& = \left(1 +o_p(1)\right)\cdot O(1)\cdot O_p\left(r_p(n)^{-1} + \sqrt{K \tau_K \log(K)/n} \right) + o_p(1)\cdot O(1) \\
& = o_p(1)
\end{align*}
with the right hand side converging to zero by consistency of $\widehat{\sigma}_{dm,k}$ for $\sigma_{dm,k}$, the fact that $\max_k \sigma_{dm,k}/\min_k \sigma_{dm,k}=O(1)$ and the rate restriction in Assumption \ref{A::Norm_I} (ii).
\end{proof}
\begin{proof}[Proof of Lemma \ref{Lem::CV}]	
By definition of $\widehat{f}_{B|X}(b|x)$ and the last display in the proof of Lemma \ref{Lem::Ident} it holds that 
	\begin{align*}
	& \int_{\mathbb{R}^2} \widehat{f}_{B|X}(b|x)f_{B|X}(b|x)db \\
	=&\int_{\mathbb{R}^2} q^K(b - \widehat  \beta(x))'\widehat Q^{-1}\widehat{\Pi}_{dm}(x) \left[\frac{1}{(2\pi)^2}\int_{\mathbb{R}^2} |t| \exp(-ib'(t,tw))\phi_{Y|X, W}(t|x,w)dtdw\right]db\\ 
	=&\int_{\mathbb{R}^2} q^K(b - \widehat  \beta(x))'\widehat Q^{-1}\widehat{\Pi}_{dm}(x) \left[\frac{1}{(2\pi)^2}\int_{\mathbb{R}^2} |t| \exp[-it(y-b'(1,w))]f_{Y|X,W}(y|x,w)dydtdw\right]db
	\end{align*}
with the last equality following from plugging in the definition for $\phi_{Y|X, W}(t|x,w)$.
Rearranging and using the definition of $V(Y,W)$ yields,	
\begin{align*}
& \int_{\mathbb{R}^2} \widehat{f}_{B|X}(b|x)f_{B|X}(b|x)db \\
=& \int_{\mathbb{R}} \Ex[V(Y,W)'\widehat Q^{-1}\widehat{\Pi}_{dm}(X) |X=x, W=w]dw
\end{align*}
which is the statement of the Lemma. 
\end{proof}
\begin{lem}\label{Lem1::Appendix}
	Let Assumptions \ref{A::Rate}- \ref{A::MLRates} be satisfied and $q^K$ be the Hermite function basis. \\
	\begin{enumerate}[(i)]
	\item If $\Pi_{dm}(x) = \Ex[T_k(W-\widehat{g}(X), Y-\widehat{m}(X,W)) |X=x]$ it holds that
	\begin{align*}
	\mynorm{\Pi_{dm}(x) - \Pi(x)}^2 = O(K^{2}\cdot n^{-2\varphi})
	\end{align*} 
	\item If $\Pi_{dm}(x) = \Ex[T_k(W, Y-\widehat{m}(X,W)) |X=x]$ then 
	\begin{align*}
	\mynorm{\Pi_{dm}(x) - \Pi(x)}^2 = O(K\cdot n^{-2\varphi})
	\end{align*}
\end{enumerate}
\end{lem}
\begin{proof}[Proof of Lemma \ref{Lem1::Appendix}]
	The proof begins with case (i). By the definitions made earlier it holds
	\begin{align*}
	& \mynorm{ \Pi_{dm}(x) -  \Pi(x)  }^2\\
	\leq & \sum_{k=1}^{K} \Ex\left[\left|T_k(W-\widehat{g}(X), Y-\widehat{m}(X,W)) - T_k(W-g(X), Y-m(X,W)) \right|^2  |X=x\right] 
	\end{align*}
	Then by the properties of complex numbers and a mean value argument  
	\begin{align*}
	& \left|T_k(W-\widehat{g}(X), Y-\widehat{m}(X,W)) - T_k(W-g(X), Y-m(X,W)) \right|^2 \\
	= &  [\text{Re}(T_k(W-\widehat{g}(X), Y-\widehat{m}(X,W))) -\text{Re}( T_k(W-g(X), Y-m(X,W)))]^2 \\
	& + [\text{Im}(T_k(W-\widehat{g}(X), Y-\widehat{m}(X,W))) -\text{Im}( T_k(W-g(X), Y-m(X,W)))]^2 \\
	=& \nabla \text{Re}(T_k)(\xi_1)'(\widehat{g}(X)-g(X), \widehat{m}(X,W)-m(X,W))^2 \\
	& +\nabla \text{Im}(T_k)(\xi_2)'(\widehat{g}(X)-g(X), \widehat{m}(X,W)-m(X,W))^2
	\end{align*}
	for $\xi_j=(W,Y)-\tau_j\cdot (g(X),m(X,W)) +(1-\tau_j)\cdot(\widehat g(X),\widehat m(X,W))$ where $\tau_j \in (0,1)$.
	
	Hence
	\begin{align*}
		& \sum_{k=1}^{K} \mynorm{ \Pi_{dm}(x) - \Pi(x)  }^2\\
	\leq & \sum_{k=1}^{K} \Ex\left[ \mynorm{\nabla\text{Re}(T_k)(\xi_1) }^2 +\mynorm{\nabla\text{Im}(T_k)(\xi_1) }^2 |X=x\right] \cdot \mynorm{\widehat{g}(x)- g(x), \Ex[\widehat{m}(X,W)-m(X,W)|X=x]}^2.
	\end{align*}
	Using the definition of $T_k$ along with the eigenfunction property of Hermite functions that 
	\begin{align*}
	\mathcal{F}q^K(-t, -tw) = \sqrt{2\pi} i^{k-1}q^K(-t, -tw)
	\end{align*}
	we obtain 
	\begin{align*}
	\text{Re}(T_k)(y,w) & = \int_{\mathbb{R}} q^K(-t, -tw)\cdot \cos(\frac{\pi\cdot K}{2}+ty) d\nu(t) \\
	\text{Im}(T_k)(y,w) & = \int_{\mathbb{R}} q^K(-t, -tw)\cdot \sin(\frac{\pi\cdot K}{2}+ty) d\nu(t)
	\end{align*}
	Let $\xi_1 = (\xi_w, \xi_y)$ then we have for the real part 
\begin{align}
 \mynorm{\nabla \text{Re}(T_k)(\xi_1) }^2  = & \frac{\partial \text{Re}(T_k)(\xi_1)}{\partial w}^2 + \frac{\partial \text{Re}(T_k)(\xi_1)}{\partial y}^2  \label{Gradient T_k} \\
  \leq  & \int_{\mathbb{R}} \frac{\partial q^K(-t, -t\xi_w)}{\partial w}^2\cos(\frac{\pi k}{2}+t\xi_y)^2 t^2 d\nu(t)  \nonumber\\
 & + \int_{\mathbb{R}} q^K(-t, -t\xi_w)^2\sin(\frac{\pi k}{2}+t\xi_y)^2t^2 d\nu(t) \nonumber  \\
 \leq & \sup_{b \in \mathbb{R}^2} \frac{ \partial q^K(b_1, b_2)}{\partial b_2}^2 \cdot \int_{\mathbb{R}} t^2 d\nu(t) + \sup_{b \in \mathbb{R}^2} q^K(b_1, b_2)^2 \cdot \int_{\mathbb{R}} t^2 d\nu(t)     \nonumber \\
 \lesssim & K \nonumber
\end{align}
which is due to the fact that the distribution $\nu$ has finite second moments, the boundedness of Hermite functions and the following property of the derivative of Hermite functions
\begin{align*}
\partial q^k(b)/\partial b = \sqrt{\frac{k}{2}}q^{k-1}(b) -  \sqrt{\frac{k+1}{2}}q^{k+1}(b). 
\end{align*} 
The argument is analogous for the imaginary part and summarizing
	\begin{align*}
& \mynorm{ \Pi_{dm}(x) - \Pi(x)  }^2\\
\lesssim & \sum_{k=1}^{K} 2\cdot K \cdot \mynorm{\widehat{g}(x)- g(x), \Ex[\widehat{m}(X,W)-m(X,W)|X=x]}^2 \\
\lesssim & K^2 \cdot  O_p(n^{-2\varphi}) 
\end{align*}
The proof for part (ii) is analogous. Here (\ref{Gradient T_k}) is only the derivative with respect to y and it immediately follows from (\ref{Gradient T_k}) that
\begin{align*}
\mynorm{\nabla \text{Re}(T_k)(\xi_1) }^2  \lesssim 1
\end{align*}
and thus
\begin{align*}
\mynorm{\Pi_{dm}(x) - \Pi(x)}^2 = O(K\cdot n^{-2\varphi})
\end{align*}
which concludes the proof. 
\end{proof}
\newpage
 \bibliography{BiB.bib}
  
\end{document}